\newif\iffull
\fulltrue

\documentclass[11pt]{article}

\usepackage{amsmath}
\usepackage{graphicx}
\usepackage[colorlinks=true, allcolors=blue]{hyperref}
\usepackage{amsfonts}
\usepackage{tcolorbox}
\usepackage{mathtools}
\usepackage{subfig}
\usepackage{adjustbox}
\usepackage{graphicx}
\usepackage{makecell}

\iffull
\usepackage{fullpage}
\usepackage{amsthm}
\newtheorem{theorem}{Theorem}
\newtheorem{definition}{Definition}
\newtheorem{lemma}{Lemma}
\newtheorem{corollary}{Corollary}
\renewcommand{\smallskip}{\medskip}
\else

\usepackage{amsthm}
\fi

\iffull
\title{Cuckoo Hashing in Cryptography:\\ Optimal Parameters, Robustness and Applications}
\author{Kevin Yeo\thanks{Google and Columbia University. \texttt{kwlyeo@google.com}}}
\else
\title{Cuckoo Hashing in Cryptography: Optimal Parameters, Robustness and Applications}
\author{Kevin Yeo}
\institute{Google and Columbia University, \texttt{kwlyeo@cs.columbia.edu}}
\fi
\date{}

\newcommand{\negl}{\mathsf{negl}}
\newcommand{\poly}{\mathsf{poly}}

\newcommand{\ignore}[1]{}

\DeclarePairedDelimiter\ceil{\lceil}{\rceil}
\DeclarePairedDelimiter\floor{\lfloor}{\rfloor}

\newcommand{\calA}{\mathcal{A}}
\newcommand{\calC}{\mathcal{C}}
\newcommand{\calG}{\mathcal{G}}
\newcommand{\calT}{\mathcal{T}}
\newcommand{\calH}{\mathcal{H}}

\newcommand{\calD}{\mathcal{D}}
\newcommand{\calR}{\mathcal{R}}
\newcommand{\calO}{\mathcal{O}}

\newcommand{\polylog}{\mathsf{polylog}}

\newcommand{\st}{\mathsf{st}}

\newcommand{\E}{\mathbb{E}}

\newcommand{\res}{\mathsf{res}}
\newcommand{\indgame}{\mathsf{IndGame_{\calA}^{\eta}}}

\newcommand{\CH}{\mathsf{CH}}
\newcommand{\cnt}{\mathsf{cnt}}

\newcommand{\sample}{\mathsf{Sample}}
\newcommand{\construct}{\mathsf{Construct}}
\newcommand{\query}{\mathsf{Query}}
\newcommand{\add}{\mathsf{Insert}}
\newcommand{\delete}{\mathsf{Delete}}
\newcommand{\params}{\mathsf{prms}}
\newcommand{\encode}{\mathsf{Encode}}
\newcommand{\decode}{\mathsf{Decode}}
\newcommand{\init}{\mathsf{Init}}
\newcommand{\genschedule}{\mathsf{Schedule}}

\newcommand{\id}{\mathsf{id}}
\newcommand{\db}{\mathsf{DB}}

\newcommand{\bv}{\mathbf{v}}

\makeatletter
\def\blfootnote{\xdef\@thefnmark{}\@footnotetext}
\makeatother

\begin{document}
\maketitle

\begin{abstract}
Cuckoo hashing is a powerful primitive
that enables storing items using small space with efficient querying. At a high level,
cuckoo hashing maps $n$ items into $b$ entries
storing at most $\ell$ items such that each item
is placed into one of $k$ randomly chosen entries.
Additionally, there is an overflow stash
that can store at most $s$ items.
Many cryptographic primitives rely upon cuckoo hashing to privately embed and query data where it is integral to ensure small failure probability when constructing cuckoo hashing tables
as it directly relates to the privacy guarantees.

As our main result, we present a more query-efficient
cuckoo hashing construction using more hash functions.
For construction failure probability $\epsilon$,
the query overhead of our
scheme is $O(1 + \sqrt{\log(1/\epsilon)/\log n})$.
Our scheme has quadratically smaller query overhead than prior works for any
target failure probability $\epsilon$.
We also prove lower bounds matching our construction.
Our improvements come from a new understanding of
the locality of cuckoo hashing failures for
small sets of items.

We also initiate the study of robust cuckoo hashing where
the input set may be chosen with knowledge of the hash functions.
We present a cuckoo hashing scheme using more hash functions with query overhead
$\tilde{O}(\log \lambda)$
that is robust against $\poly(\lambda)$ adversaries.
Furthermore, we present lower bounds showing that this construction is tight and that extending previous approaches of large stashes or entries
cannot obtain robustness except with $\Omega(n)$ query overhead.

As applications of our results, we obtain improved constructions
for batch codes and PIR.
In particular, we present the most efficient
explicit batch code and blackbox reduction from single-query PIR to batch PIR.
\end{abstract}

\section{Introduction}
\label{sec:intro}
Cuckoo hashing, introduced by Pagh and Rodler~\cite{pagh2004cuckoo}, is
a powerful tool that enables embedding data from a very large
universe into memory whose size is linear in the total size
of the data while enabling very efficient retrieval. In more detail, the original cuckoo hashing
scheme enables
taking a set of $n$ items from a universe $U$ and stores them
into approximately $2n$ entries such that querying any item $x$
requires searching only two entries.
A huge advantage of cuckoo hashing is that the storage overhead is optimal up to a small constant factor and independent
of the universe size $|U|$ while query overhead
(the number of possible locations for any item) is $O(1)$. Due to the power of cuckoo hashing,
it has found usage in a wide range of applications such as high-performance hash tables~\cite{zhang2015mega,breslow2016horton}, databases~\cite{polychroniou2015rethinking}, caching~\cite{fan2013memc3}
and cuckoo filters~\cite{fan2014cuckoo}.
Furthermore, many follow-up works have studied further properties
and variants of cuckoo hashing including~\cite{devroye2003cuckoo,fotakis2005space,kutzelnigg2006bipartite,dietzfelbinger2007balanced,arbitman2009amortized,frieze2009analysis,kirsch2010more,fountoulakis2013insertion,aumuller2014explicit,minaud2020note}.

\smallskip\noindent{\bf Cuckoo Hashing in Cryptography.}
One important area where cuckoo hashing has found wide usage is
cryptography.
Cuckoo hashing is a core component of many cryptographic primitives such as
private information retrieval (PIR)~\cite{angel2018pir,demmler2018pir,ali2021communication},
private set intersection (PSI)~\cite{pinkas2015phasing,chen2017fast,chen2018labeled,pinkas2018efficient,pinkas2020psi}, symmetric
searchable encryption (SSE)~\cite{patel2019mitigating,bossuat2021sse}
and oblivious RAM (ORAM)~\cite{pinkas2010oblivious,goodrich2011privacy,patel2018panorama,asharov2020optorama,hemenway2021alibi}.

A common method used in cryptographic primitives is to leverage cuckoo hashing to {\em privately} embed data while enabling
efficient queries when necessary. For example, suppose
one party has a database of identifier-value pairs
$\{(\id_1,v_1),\ldots,(\id_n,v_n)\}$ that it wishes
to outsource to another potentially untrusted third party for storage. As the third party is untrusted, the data must
be outsourced in a private manner such that the third party cannot
see the data in plaintext. For utility, the data
owner should still be able to query and retrieve
certain values efficiently.
To do this, many works leverage cuckoo hashing with two
modifications. First, the underlying random hash functions
are replaced with cryptographic hash functions
(typically, pseudorandom functions).
Secondly, the contents of the resulting cuckoo hash tables
are encrypted in some manner such as standard IND-CPA encryption.
The keys for the cryptographic hash function and encryption
are typically kept by the data owner to ensure privacy.
To query, the data owner executes the
cuckoo hashing query algorithm using the private keys when necessary to retrieve all possible locations
for a queried item.
The above does not comprehensively cover all usages
of cuckoo hashing in cryptography, but was elaborated to provide readers with intuition on an example usage of cuckoo hashing
to privately embed and query data.

To our knowledge, all cryptographic applications use cuckoo hashing where all $n$ items are provided
ahead of time to construct a hash table that is not modified
in the future.
However, this does not preclude the usage of cuckoo hashing for
cryptographic primitives where data
may be updated frequently (such as oblivious RAM).
Throughout our work, we will focus
on the setting of constructing a static cuckoo hash table and ignore features that
enable inserting items.
See Section~\ref{sec:tables} for more discussion on our choice of abstraction.

\smallskip\noindent{\bf Failure Probability and Adversarial Advantage.}
Requirements of cuckoo hashing for usage in cryptography differs
significantly from other applications.
In cuckoo hashing, there is a non-zero
probability that it is impossible to allocate a set of $n$
items using the sampled random hash functions. We will
refer to this as the {\em construction failure probability}. For
standard cuckoo hashing, it was shown
that the failure probability is $1/\poly(n)$. Without privacy concerns,
one can simply
sample new random hash functions and repeat the construction
algorithm to handle the failure. As the failure probability is $1/\poly(n)$,
this would increase the expected running time
of construction by a minimal amount.

For cryptography, the failure probability has much larger
implications with respect to privacy and adversarial advantage.
Suppose we consider a cuckoo hashing instantiation
with failure construction probability $\epsilon$ over the randomly sampled hash function $\calH$. This means
that, if we choose a uniformly random set of $n$ items,
the set of $n$ items cannot be allocated
correctly according to $\calH$ with probability $\epsilon$.
In other words, an $\epsilon$-fraction of inputs
will behave differently from the remaining inputs. At a high
level, an adversary can leverage this property to compromise
privacy of the inputs when $\epsilon$ is too large.
Suppose the adversary picks two input sets $S_1$ and $S_2$.
An ideal cryptographic protocol would pick $\eta \in \{0,1\}$
at random and execute with $S_\eta$ as input. Given
the transcript of the protocol, the adversary
should not be able to guess $\eta$ with probability better than $1/2 + \negl(\lambda)$ probability.
If the adversary picks $S_1$ and $S_2$ at random,
it is not hard to see that exactly
one of $S_1$ or $S_2$ will fail to be allocated
by cuckoo hashing with probability approximately $O(\epsilon)$.
As the transcripts will be different when cuckoo hashing
fails to construct the hash table, the adversary
has advantage approximately $1/2 + \epsilon$ of guessing
$\eta$. So, it is essential that the construction failure probability $\epsilon$ is negligible to ensure privacy of the embedded data.
In fact, prior works have shown insecurity
of protocols when $\epsilon = 1/\poly(n)$ such as~\cite{kushilevitz2012security}.
Therefore,
standard cuckoo hashing with $1/\poly(n)$ failure cannot be used in most cryptographic
applications.
Instead, we must come up with new cuckoo hashing schemes with
negligible failure for usage in cryptography.

\smallskip\noindent{\bf Cuckoo Hashing with Negligible Failure.}
To systematically study cuckoo hashing, we will consider several
tunable parameters: number of hash functions $k$, number of entries $b$,
size of each entry $\ell$ and size of the overflow stash $s$.
The main table will consist of $b$ entries that can each store up to $\ell$ items.
Each item is randomly assigned to $k$ different entries using the hash functions where the item may be allocated.
Additionally, there is an overflow stash of size $s$ to store
any items that cannot fit in the main table.
We denote the {\em query overhead} as the total number of possible
locations that need to be checked when querying an item.
With these parameters, the query overhead is exactly $k\ell + s$
that checks all $\ell$ locations in each of the $k$ entries as well
as the stash.
To our knowledge, the above encapsulates
all variants of cuckoo hashing used in cryptography.

Obtaining negligible (or even zero) failure in cuckoo hashing is trivial.
For example, if we set the stash size to be $s = n$, no failures will
ever occur. The real challenge is obtaining negligible failure while
keeping small query overhead that directly relates to the efficiency
of the cryptographic application. While the above example
obtains zero failure, the resulting query overhead is
$k\ell + s = O(n)$.

It is an important question to study the query overhead of cuckoo hashing with negligible failure
due to its heavy usage in cryptographic primitives.
To date, there are two constructions that obtain
the smallest query overhead.
The first is cuckoo hashing with
a large stash introduced by Kirsch, Mitzenmacher and Wieder~\cite{kirsch2010more}. By picking stash size $s = O(1 + \log(1/\epsilon)/\log n)$, it can be proven that cuckoo hashing
will fail only with $\epsilon$ probability~\cite{aumuller2014explicit}.
Another approach considers entries
that can store a large number of items $\ell$. For $\ell = O(1 + \log(1/\epsilon)/ \log n)$, it has been
shown that failure
probability $\epsilon$ may be achieved by Minaud and Papamanthou~\cite{minaud2020note}.
A final approach
considers standard cuckoo hashing that utilizes large-scale experiments to estimate the failure probability~\cite{chen2017fast,angel2018pir} without
providing any provable failure bounds.
The above approaches all obtain negligible failure probabilities in different ways. It is unclear that the above approaches
are most efficient way to obtain negligible failure leading us to the following question:
\begin{center}
{\em What is the smallest query overhead achievable with provably\\ negligible failure probability for cuckoo hashing?}
\end{center}
\noindent

As our major result, we present
an approach with provable failure bounds that is quadratically smaller query overhead than previous constructions. We also prove lower bounds matching our construction.

\begin{figure*}[tb]
\centering
\begin{adjustbox}{width=1.25\linewidth, center}
%\small
\begin{tabular}{| l | c | c | c | c | c | c |}
        \hline
       & \makecell{\bf Hash Functions $k$} & \makecell{\bf Entry Size $\ell$} &  \makecell{\bf Entries $b$} & \makecell{\bf Stash Size $s$} & \makecell{\bf Failure $\epsilon$} & \makecell{\bf Query Overhead}
\\ \hline
    \makecell[l]{Cuckoo Hashing~\cite{pagh2004cuckoo}} & $2$ & $1$ & $O(n)$ & 0 & $1/n^{O(1)}$ & $O(1)$ \\ \hline
    \makecell[l]{Large-Sized Entries~\cite{dietzfelbinger2007balanced}} & $2$ & $O(1)$ & $(1+\alpha)n/\ell$ & $0$ & $1/n^{O(1)}$ & $O(1)$\\ \hline
    \makecell[l]{Large-Sized Entries~\cite{minaud2020note}} & $2$ & $O(1 + \log(1/\epsilon)/\log n)$ & $O(n/\ell)$ & $0$ & $\epsilon$ & $O(1 + \log(1/\epsilon)/\log n)$\\ \hline
    \makecell[l]{Constant-Sized Stash~\cite{kirsch2010more}} & $2$ & $1$ & $O(n)$ & $O(1)$ & $1/n^{O(s)}$ & $O(1)$ \\ \hline
    \makecell[l]{Large-Sized Stash~\cite{aumuller2014explicit}} & $2$ & $1$ & $O(n)$ & $O(1 + \log(1/\epsilon) / \log n)$ & $\epsilon$ & $O(1 + \log(1/\epsilon)/\log n)$\\ \hline
    \makecell[l]{More Hash Functions~\cite{fotakis2005space}} & $O(1 + \log(1/\epsilon)/\log n)$ & $1$ & $O(n)$ & $0$ & $\epsilon$ & $O(1 + \log(1/\epsilon)/\log n)$\\ \hline
    \hline
    
    \makecell[l]{Our Work} & $\boldsymbol{O(1 + \sqrt{\log(1/\epsilon)/\log n})}$ & $1$ & $O(n)$ & $0$& $\epsilon$ & $\boldsymbol{O(1 + \sqrt{\log(1/\epsilon)/\log n})}$\\ \hline
\end{tabular}
\end{adjustbox}
\caption{Comparison table of known cuckoo hashing instantiations.
The query overhead $k\ell + s$ is the number of locations to search when retrieving an item.
}
\label{table:ch_compare}
\iffull
\else
\vspace*{-6mm}
\fi
\end{figure*}

\smallskip\noindent{\bf Adversarial Robustness.}
In the prior discussions, we overlooked
a subtle, but important, assumption
used in cuckoo hashing.
The failure probabilities assumed that the chosen inputs
were independent of the sampled
random hash functions.
Instead, if we assume
that an adversary is given the
hash functions to choose the input
set for cuckoo hashing, the
previous
construction failure bounds
no longer apply.
In many settings, it is natural
that the adversary has knowledge
of the random hash functions.
Two such examples include
if the adversary may view or control
the randomness in the system or
if the hash functions need to be published publicly for use by multiple parties.
\ignore{
For example, the latter is a requirement
of re-usable batch PIR schemes (see Section~\ref{sec:reusable_batch_pir} for more details).}

We initiate the study of 
{\em robust cuckoo hashing}
that
provide negligible construction failure
probabilities for
inputs chosen adversarially with knowledge of the
hash functions.
This
leads to the following natural question:
\begin{center}
{\em What is the smallest query overhead achievable for robust cuckoo hashing?}
\end{center}
In our work, we present
constructions for robust cuckoo hashing
with optimal query overhead that match our lower bounds.

\subsection{Our Contributions}

\noindent{\bf Improved Query and Failure Trade-offs.}
As our major result, we present
a new cuckoo hashing construction
that achieves better trade-offs between query overhead
and failure probabilities.
To obtain failure probability
$\epsilon$, we prove it suffices to
use $k = O(1 + \sqrt{\log(1/\epsilon)/\log n})$
hash functions with $b = O(n)$ entries of size $\ell = 1$
and no stash, $s = 0$.
Therefore, the resulting query overhead is
$O(1 + \sqrt{\log(1/\epsilon)/\log n})$.
If we restrict to the case when $\epsilon \le 1/n$, we get that $\log(1/\epsilon) \ge \log n$.
Then, we can drop the case when $\epsilon$ is too large
and simply state that the number of hash functions and query overhead is $O(\sqrt{\log(1/\epsilon)/\log n})$.
We will assume $\epsilon \le 1/n$ through the rest of the section for convenience.

For any target failure probability $\epsilon \le 1/n$,
our construction requires quadratically smaller query overhead
than any prior known scheme.
To date, the best query overhead achievable was
$O({\log(1/\epsilon)/\log n})$
of cuckoo hashing schemes
instantiated with large stashes~\cite{aumuller2014explicit}
or larger entries~\cite{minaud2020note}.
We also prove matching lower bounds
showing that our construction is optimal and provides
the smallest query overhead across all possible
instantiations using the four considered parameters.
Our results
and comparisons to prior work
can be found in Figure~\ref{table:ch_compare}.

\begin{figure*}[tb]
\centering
\begin{adjustbox}{width=1.1\linewidth, center}
%\small
\begin{tabular}{| l | c | c | c | c | c | c |}
        \hline
       & \makecell{\bf Hash Functions $k$} & \makecell{\bf Entry Size $\ell$} &  \makecell{\bf Entries $b$} & \makecell{\bf Stash Size $s$} & \makecell{\bf Robustness} & \makecell{\bf Query Overhead}
\\ \hline
    \makecell[l]{Our Work} & $\boldsymbol{O(f(\lambda)\log \lambda), f(\lambda) = \omega(1)}$ & $1$ & $O(n)$ & $0$ & $(\lambda, \negl(\lambda))$ & $\boldsymbol{O(f(\lambda)\log \lambda)}$ \\ \hline\hline

   \makecell[l]{Our Work} &
   $\boldsymbol{\omega(\log \lambda)}$ & $o(n)$ & $O(n/\ell)$ & $o(n)$ & $(\lambda, 1/2)$ & $\boldsymbol{\omega(\log \lambda)}$\\ \hline
   \makecell[l]{Our Work} &
   $O(\log \lambda)$ & $\boldsymbol{\Omega(n)}$ & $O(n/\ell)$ & $o(n)$ & $(\lambda, 1/2)$ & $\boldsymbol{\Omega(n)}$\\ \hline
    \makecell[l]{Our Work} &
   $O(\log \lambda)$ & $o(n)$ & $O(n/\ell)$ & $\boldsymbol{\Omega(n)}$& $(\lambda, 1/2)$ & $\boldsymbol{\Omega(n)}$\\ \hline
\end{tabular}
\end{adjustbox}
\caption{A table of bounds for cuckoo hashing parameters where $(\lambda,\epsilon)$-robustness
means that any adversary running in probabilistic $\poly(\lambda)$ time cannot cause a construction failure with probability greater than $\epsilon$.}
\label{table:robust_ch_compare}
\iffull
\else
\vspace*{-6mm}
\fi
\end{figure*}

\smallskip\noindent{\bf Robust Cuckoo Hashing.}
We also study the best query overhead
achievable for robust cuckoo hashing.
We say cuckoo hashing is $(\lambda,\epsilon)$-robust if any
adversary running in $\poly(\lambda)$ time cannot find
an input set that will incur a construction failure
with probability greater than $\epsilon$.
We present a $(\lambda, \negl(\lambda))$-robust construction with
$k = O(f(\lambda) \log \lambda)$ hash functions, 
for any super-constant function $f(\lambda) = \omega(1)$, with $b = O(n)$
entries of size $\ell = 1$ and no stash, $s = 0$.
Therefore, our construction has query overhead
$O(f(\lambda) \log\lambda)$.
For polynomial time adversaries with $\lambda = n$,
the resulting query overhead is essentially $\tilde{O}(\log n)$.

We also prove matching lower bounds showing our construction
is optimal. If we restrict ourselves to sub-linear query overhead constructions with sub-linear sized entries and stash, $s = o(n)$ and
$\ell = o(n)$, we show that the number of hash functions must be
$\omega(\log \lambda)$.
The above also shows that the only way to obtain sub-linear query overhead
is using a large number of hash functions. For example, any
instantiation with $k = O(\log \lambda)$ requires that $s + \ell = \Omega(n)$ meaning that query overhead is linear.
All our results are summarized
in Figure~\ref{table:robust_ch_compare}.

\smallskip\noindent{\bf Applications.}
Using our new cuckoo
hashing constructions with large number of hash functions, we present improved
constructions for several primitives:
\begin{itemize}
    \item {\em Probabilistic Batch Codes (PBC).} PBCs are a primitive that
    aim to encode a database $D$ of $n$ entries into $m$ buckets such that any subset
    of $q$ entries may be retrieved efficiently by accessing at most
    $t$ codewords from each bucket.
    The total number of codewords is denoted by $N$ and the goal is to maximize the rate $n/N$ and to keep
    the buckets $m$ as close to $q$ as possible. Utilizing our efficient
    cuckoo hashing scheme, we obtain a PBC with rate $O(\sqrt{\lambda/\log\log n})$, $m = O(q)$ and $t=1$
    with error probability $2^{-\lambda}$. Our construction has quadratically better rate than prior provable PBCs.
    \item {\em Robust PBCs.} We also initiate the study of {\em robust PBCs}
    where the subset of $q$ entries may be chosen by an adversary with
    knowledge of the system's randomness. By leveraging our robust cuckoo
    hashing constructions, we obtain robust PBCs with rate $O(\log\lambda)$
    while ensuring that a PPT adversary cannot find an erring input
    except with probability $\negl(\lambda)$. To our knowledge, this is the most efficient explicit robust PBC to date.
    All prior perfect PBCs with zero error are either less efficient
    or are non-explicit.
    \item {\em Single-Query to Batch PIR.}
    Private information retrieval (PIR) considers the setting where a client wishes to privately retrieve an entry
    from an array stored on a server. Batch PIR is the extension where the
    client wishes to retrieve a subset of $q$ entries at once. A standard
    way to build batch PIR is to compose a single-query PIR with a (probabilistic) batch code (see~\cite{IKO04,angel2018pir}). To our knowledge, these compositions are the most concretely
    efficient ways to build batch PIRs to date. Using our new PBCs, we obtain
    a more asymptotically efficient blackbox reduction from single-query PIR to batch PIR.
    \item {\em Re-usable Batch PIR.}
    We also consider re-usable batch PIR where the server must efficiently
    handle multiple sequential queries. The standard reduction from single-query to batch PIR requires the server to encode a database
    using a PBC where the construction failure probability becomes
    the error probability. In the case of multiple queries (i.e. re-usable protocols),
    the hash functions must be made public to all parties. As a result,
    adversarially chosen queries can be made to fail with very high probability.
    To solve this problem, we utilize robust PBCs to construct a re-usable batch PIR that guarantees
    a PPT adversary will be unable to find erring query subsets.
    \item {\em Other Applications.}
    We also show that our new cuckoo hashing schemes
    may also be applied to improve other primitives including
    private set intersection, encrypted search, vector oblivious linear evaluation and batch PIR with private preprocessing.
\end{itemize}

\subsection{Related Works}

\noindent{\bf Cuckoo Hashing Variants.}
Following the seminal work of Pagh and Rodler~\cite{pagh2004cuckoo},
there have been many follow-up works studying cuckoo hashing. We will focus on
the variants that enable negligible failure probability.
Cuckoo hashing where entries can store $\ell > 1$ items was studied
in~\cite{dietzfelbinger2007balanced,kirsch2010more,minaud2020note}.
%where negligible probability could be achieved for super-constant $\ell$.
A variant of cuckoo hashing with an overflow stash
that may store $s \ge 1$ items was studied in~\cite{kirsch2010more,aumuller2014explicit,minaud2020note,cryptoeprint:2021/447}.
%where super-constant $s$ could obtain negligible failure.
Finally, cuckoo hashing with $k > 2$ hash functions
was studied in~\cite{fotakis2005space,dietzfelbinger2010tight}.

\smallskip\noindent{\bf Cuckoo Hashing in Cryptography.}
Many prior works have relied heavily upon
cuckoo hashing to build many primitives including
private information retrieval (PIR)~\cite{angel2018pir,demmler2018pir,ali2021communication},
private set intersection (PSI)~\cite{pinkas2015phasing,chen2017fast,chen2018labeled,pinkas2018efficient,mohassel2020fast,pinkas2020psi,duong2020catalic}, symmetric
searchable encryption (SSE)~\cite{patel2019mitigating,bossuat2021sse}, oblivious RAMs~\cite{goodrich2011privacy,patel2018panorama,asharov2020optorama,hemenway2021alibi},
history-independent data structures~\cite{naor2008history}
and hardness-preserving reductions~\cite{berman2019hardness}.
We also point readers to references therein for more prior works.

\smallskip\noindent{\bf Adversarial Robustness.}
Similar notions of
adversarial robustness
has been studied in prior works
outside of cuckoo hashing
where it is assumed the adversary has knowledge
of the system's randomness.
Some examples of prior works include sketching~\cite{mironov2011sketching,hardt2013robust},
streaming~\cite{ben2020framework},
probabilitic data structures~\cite{naor2015bloom,clayton2019probabilistic,filic2022adversarial}
and property-preserving hash functions~\cite{boyle2019adversarially,fleischhacker2021property,holmgren2022nearly}.

\smallskip\noindent{\bf Other Hashing Schemes.}
We note that there are other schemes beyond cuckoo hashing
to obtain efficient allocations. For example, the
``power-of-two'' choice paradigm inserts items into the
least loaded of two bins~\cite{azar1994balanced,richa2001power}.
A variant with a stash to obtain negligible failure
was presented in~\cite{patel2019storage}.
Examples of other schemes include simple tabulation~\cite{patracscu2012power} and
multi-dimensional balanced allocations~\cite{asharov2016searchable}.

\section{Technical Overview}
\label{sec:overview}
In this section, we present a technical overview
for our improved cuckoo hashing constructions.
To start, we consider a failed modification 
that provides insights into cuckoo hashing failures.
Equipped with these insights, we derive both our improved constructions and matching lower bounds.

\smallskip\noindent{\bf Insights from a Failed Construction.}
A standard instantiation of cuckoo hashing used in cryptography
including oblivious RAM~\cite{patel2018panorama,asharov2020optorama} and encrypted search~\cite{patel2019mitigating} uses a stash of size $s = O(1 + \log(1/\epsilon)/\log n)$, $k = 2$ hash functions and $b = O(n)$ entries of size $\ell = 1$
to obtain $\epsilon$ failure probability.
A straightforward modification may be to try and recursively apply cuckoo hashing
on the overflow items in the stash of logarithmic size. In particular, we can try
to build a second cuckoo hashing table with $b = O(n)$ entries of size $\ell = 1$ and $k = 2$ hash functions for all items in the stash. This seems like
a simpler task as there are only a logarithmic number of items to
place in a table with $O(n)$ entries. If the resulting stash becomes zero, we would make the query overhead to be $O(1)$ as only $2k$ locations need to be checked. Even if the resulting stash was smaller, we could apply this technique recursively to obtain smaller query overhead. Unfortunately, this approach is unsuccessful as the second
table also requires the stash size to be $s = O(1 + \log(1/\epsilon)/\log n)$
if we re-do the failure analysis in~\cite{aumuller2014explicit}.

However, we develop the following insights from the failed construction.
The failure of cuckoo hashing seems to be fairly localized and only depend
on a small sets of items. We observe that handling
the $O(1 + \log(1/\epsilon)/\log n)$ overflow items in the second table
required
preparing for a similar number of overflows as the first table with $n$ items. On the positive side,
once a cuckoo hashing construction can handle overflows for a small set of items, it seems
that this will scale to a large set of items easily.
In other words, if one designs
a cuckoo hashing scheme that can handle small sets of items, that scheme should also be able
to handle larger sets of items.
We see this as the first table
requires the same asymptotic stash size for $n$ items as the second table for a logarithmic number of items.
We rely on these insights for our constructions and lower bounds.

\smallskip\noindent{\bf Our Improved Construction.}
To obtain our improved construction, we will heavily utilize our observation
that cuckoo hashing failures are very local to sets of small sizes.
We need to find a way to ensure that small sets of items will not cause significant
overflows requiring large entries or stashes. One way to do this is to
split the table of $b$ entries into $k$ disjoint sub-table of $b/k$ entries.
Each of the $k$ hash functions will be used to pick a random entry from each
of the $k$ sub-tables. Using the approach, we deterministically
guarantee that any set of $k\ell$ items will
never cause an overflow. In the worst case, all $k\ell$ items will hash to the same
$k$ entries. However, there are $k\ell$ locations to allocate all $k\ell$ items.

By considering larger values of $k$ and/or $\ell$, we could increase the size of
the sets that are handled deterministically by disjoint sub-tables.
To figure out which is the better choice,
we consider one more item that needs to be allocated even if the prior $k$ items were allocated to the same $k$ entries. This item would cause an overflow
only if it was also allocated to the same $k$ entries that occurs with probability at most
$(1/(b/k))^k = (k/b)^k$ that decreases exponentially
for larger values of $k$. Therefore, a promising approach seems to be use
$k$ disjoint sub-tables along with a larger number of hash functions $k$.

To analyze the failure probability of our new construction, we utilize prior
connections (used in~\cite{fotakis2005space,aumuller2014explicit} for example) between the cuckoo hashing and matchings in biparite graphs.
At a high level, we construct a bipartite graph with $n$ left nodes representing
the $n$ inserted items and $b\ell + s$ right nodes representing table locations
in the $b$ entries of size $\ell$ and the stash. There is an edge if and only
if an item may be allocated to the potential location as denoted by the $k$ hash functions.
We note that each left node has degree exactly $k\ell + s$ corresponding to the
$k$ assigned entries by the $k$ hash functions and the overflow stash.
Finally, we note that there exists an allocation of the $n$ items as long
as there exist a perfect left matching in the bipartite graph.
To analyze the existence of a perfect left matching, we utilize
Hall's Theorem that states such a perfect left matching exists if and only if
every the neighborhood of every subset of left nodes contains at least as many
right nodes. Using this analysis, we are able to show that $k = O(1 + \sqrt{\log(1/\epsilon)/\log n})$
hash functions, $b = O(n)$ entries of size $\ell = 1$ and no stash $s = 0$ is
sufficient to obtain $\epsilon$ failure probability.
As a result, this construction has quadratically smaller query overhead
than prior instantiations.

We also note that prior work by Fotakis {\em et al.}~\cite{fotakis2005space}
had studied cuckoo hashing with many hash functions.
However, they considered a single shared table where each of the $k$
hash functions picks any of the $b$ entries at random.
Their analysis showed that $k = O(1 + \log(1/\epsilon)/\log n)$ hash functions were
required that we will later
show is necessary for a single shared table.

\smallskip\noindent{\bf Deriving a Lower Bound.}
Next, we attempt to derive a lower bound using the observation that it seems equally challenging
to handle a small number of items as it is to handle a large number of items.
By looking at a small number of items, we can prove a very simple lower bound. Consider any instantiation
with $k \ge 1$ hash functions, $\ell \ge 1$ entry sizes, $b = O(n)$ entries and
a overflow stash of $s \ge 0$ items. First, we consider the case where we consider
$k$ disjoint sub-tables as done in our construction.
We analyze the probability that $k\ell + s + 1$ items
are allocated to the same $k$ entries. This would mean that $s + 1$ items are assigned
to the stash that would incur a construction failure. After doing exercises in probability,
one can obtain the following lower bound of $k^2 \ell + ks = \Omega(\log(1/\epsilon)/\log n)$.
Immediately, we see that our above construction with $k = O(\sqrt{\log(1/\epsilon)/\log n})$ hash functions and
$\ell = 1$ and $s = 0$ is optimal.

Finally, we can also see that it is critical to utilize disjoint sub-tables to
systematically handle smaller sets of items. By re-doing the above analysis
for a single shared table
where the $k$ hash functions pick any of the $b$ entries at random, we can immediately
obtain the lower bound $k = \Omega(\log(1/\epsilon)/\log n)$ when $b = O(n)$ and $s = \ell = O(1)$. In other words, the analysis in~\cite{fotakis2005space} for the same setting of a single, shared table
is optimal.

\smallskip\noindent{\bf Extending to Robustness.}
Lastly, we extend our results to the case of robust cuckoo hashing where
an adversary may have knowledge of the underlying randomness and hash functions
used in the cuckoo hashing scheme. In particular,
we want cuckoo hashing instantiations with small failure even if the adversary
chooses bad input sets using the hash functions. 

For our construction, we use the same insights to use a large number of hash functions $k$
along with $k$ disjoint sub-tables. We modify the analysis for robustness. We start
with an adversary that is limited to at most $Q$ evaluations of the underlying hash functions.
Without loss of generality, we also assume the $n$ submitted items by the adversary
were also evaluated meaning at most $Q + n$ hash evaluations are performed.
Afterwards, we consider the same bipartite graph where the left nodes consists of
the at most $Q + n$ items for which the adversary knows the hash evaluations.
Using a similar analysis, we analyze the probability that any subset of left nodes
would violate the requirements of Hall's theorem. We get
that $k = O(\log(Q/\epsilon))$ hash functions is sufficient to obtain
$\epsilon$ failure probability when $b = O(n)$, $\ell = 1$ and $s = 0$.
For $\poly(\lambda)$ time adversaries, we extend this to show that $k = O(f(\lambda)\log\lambda + \log(1/\epsilon))$, for any $f(n) = \omega(1)$, suffices
to obtain $\epsilon$ failure probability.
In other words, there exists a robust cuckoo hashing instantiation with roughly
logarithmic query overhead.

One may notice that the above analysis may seem too pessimistic. It computes the probability whether
there exists a subset amongst the $Q + n$ items that would violate Hall's theorem. However,
finding such a violating set is non-trivial and may not be efficiently computable by an adversary.
Nevertheless, we show that this construction is optimal by presenting a matching query overhead lower bound $k = \omega(\log \lambda)$ for robust cuckoo hashing against $\poly(\lambda)$ adversaries.
To do this, we show it suffices to consider an adversary
that simply attempts to find a set of $n$ items that all allocate into the first half
of each of the $k$ disjoint sub-tables.
Restricting to sub-linear query overhead with sub-linear sized entries
and stashes, $\ell = o(n)$ and $s = o(n)$,
we show that $k = \omega(\log \lambda)$ hash functions is necessary. In other words,
if the number of hash functions is even slightly smaller at $k = O(\log n)$, then it must be that $\ell + s = \Omega(n)$ meaning that query
overhead must be linear.
Therefore, the only way to obtain efficient query overhead
for robust cuckoo hashing is a large number of hash functions.

\section{Definitions}
\label{sec:def}
\subsection{Random Hash Functions}
\label{sec:random_functions}

We start by presenting definitions of the
various random hash functions that may
be utilized by our cuckoo hashing instantiations.

\begin{definition}[$t$-Wise Independent Hash Functions]
Let $\calH$ be a family of hash functions where
$H: \calD \rightarrow \calR$ for every hash function $H \in \calH$. $\calH$ is a $t$-wise independent hash
family if, for any distinct $x_1,\ldots,x_t \in \calD$ and any $y_1,\ldots,y_t \in \calR$, the following is true:
$$
\Pr[H(x_1) = y_1, \ldots, H(x_t) = y_t] = {|\calR|^{-t}}.
$$
\end{definition}

Next, we define a variant of hash families
that behave like a $t$-wise independent hash family
on every set of $t$ inputs except with probability
$\epsilon$. In other words, we can treat such hash families
as $t$-wise independence except with probability $\epsilon$
over the choice of the hash function (not the input set).

\begin{definition}[$(t,\epsilon)$-Wise Independent Hash Functions]
Let $\calH$ be a family of hash functions where
$H: \calD \rightarrow \calR$ for every hash function $H \in \calH$. $\calH$ is a $(t,\epsilon)$-wise independent hash
family if, for any distinct $S = \{x_1,\ldots,x_t\} \in \calD$
and any $y_1,\ldots,y_t \in \calR$, the following is true:
\begin{itemize}
    \item There exists an event $E_S$ such that $\Pr[\overline{E_S}] \le \epsilon$.
    \item $\Pr[H(x_1) = y_1, \ldots, H(x_t) = y_t : E_S] = {|\calR|^{-t}}$.
\end{itemize}
\end{definition}

These weaker variants of hash functions
have been studied in the past (such as~\cite{pagh2008uniform,aumuller2014explicit}). It has been
shown they can be constructed using lower independence
hash functions and faster evaluation.
Regardless for both $t$-wise and $(t,\epsilon)$-wise, for usable regimes of $\epsilon$, independent hash
functions require $\Omega(t)$ storage
to represent a random hash function from the
hash family.

The cost of storage for explicit random
hash functions is large. Instead, we can resort
to cryptographic assumptions to
obtain random functions.
The above definition of $(t,\epsilon)$-wise independent
hash functions considers the failure event $E_S$
from a statistical viewpoint. Instead, we could also consider
computational asssumptions
where we can use pseudorandom functions (PRFs)
that are indistinguishable from random
functions and may be represented succinctly.

\begin{definition}[Pseudorandom Functions]
For a security parameter $\lambda$, a deterministic function $F: \{0,1\}^s \times \{0,1\}^n \rightarrow \{0,1\}^m$ is a $\lambda$-PRF if:
\begin{itemize}
    \item Given any $k \in \{0,1\}^s$ and $x \in \{0,1\}^n$, there exists a polynomial time algorithm $A$ to compute $F(k, x)$.
    \item For any PPT adversary $\calA$,
    $
    \left|\Pr\left[\calA^{F(k, \cdot)}(1^\lambda) = 1 \right] - \Pr\left[\calA^{R(\cdot)}(1^\lambda) = 1 \right] \right| \le \negl(\lambda)
    $
    where $k$ is drawn randomly from $\{0,1\}^s$
    and $R$ is a random function from $\{0,1\}^n \rightarrow \{0,1\}^m$.
\end{itemize}
\end{definition}

We note that
PRFs can be viewed as $(\poly(\lambda),\negl(\lambda))$-wise independent
hash functions against all
PPT adversaries. For
$t = \poly(\lambda)$, a $\lambda$-PRF
is a $(t, \negl(\lambda))$-wise independent hash function
in the view of a PPT adversary.
This is stronger as it applies to any
$t = \poly(\lambda)$ as opposed
to a fixed $t$.

Finally, we can also consider the ideal random oracle model (ROM) where there is an oracle
$\calO$ that always output random values $\calO(x)$
for each input $x$. For any repeated inputs, $\calO$
returns the same consistent output.

\begin{definition}[Random Oracle Model]
In the random oracle model (ROM), there exists an
oracle $\calO$ such that for new input $x$, $\calO(x)$ is uniformly random. For any repeated input, the same output
$\calO(x)$ is returned.
\end{definition}

\noindent{\bf Choosing the Hash Function.} Throughout our work, we will assume that the hash function outputs
are random.
To do this,
we can
choose to instantiate our hash functions
using any of the above options.
The main differences are that
$t$-wise and $(t,\epsilon)$-wise independence
require larger storage while
PRFs and ROM require stronger underlying assumptions. In any of our results, we can
switch between the above choices
using different assumptions and storage costs.

\subsection{Hashing Schemes}
We start by defining the notion of hashing 
schemes. At a high level, the goal of a hashing
scheme is to take $n$ identifier-value pairs
$\{(\id_1,v_1),\ldots,(\id_n, v_n)\}$ with $n$
distinct identifiers from a potentially large identifier
universe $U^\id$ and allocate them into
a hash table $T$ whose size
depends only on $n$ and not the universe size $|U^\id|$.
Throughout the rest of our work, we will refer
to an identifier-value pair $(\id, v)$ as an {\em item}.
Furthermore, the hash table $T$ should enable
efficient queries for any identifier $\id \in U^\id$.
We will focus
on the setting of constructing
hash tables from an input data set
as this closely corresponds to the
usage of cuckoo hashing in cryptography.

\begin{definition}[Hashing Schemes]
A hashing scheme for size $n$, identifier universe $U^\id$ and value universe $U^V$ consists
of the following efficient algorithms:
\begin{itemize}
    \item $\calH \leftarrow \sample(1^\lambda)$: A sampling algorithm that is given a security parameter $\lambda$ as input and returns
    a set of one or more hash functions $H$.
    \item $T \leftarrow \construct(H, X)$: A construct algorithm that is given the set of hash functions $\calH$ and a set $X = \{(\id_1,v_1), \ldots, (\id_n,v_n)\} \subseteq U^\id \times U^V$ of items such that $\id_i \ne \id_j$ for all $i \ne j \in [n]$ and returns
    a hash table $T$ allocating $X$ or $\perp$ otherwise.
    \item $v \leftarrow \query(H, T, \id)$: A query algorithm that is given the set of hash functions $\calH$, the hash table $T$ and an identifier $\id$ and returns a value $v$ if $(\id, v) \in X$
    or $\perp$ otherwise.
\end{itemize}

\ignore{
Additionally, hashing schemes may optionally
enable the following efficient algorithms:
\begin{itemize}
    \item $T' \leftarrow \add(\calH, T, (\id, v))$: An insert algorithm that is given the set of hash functions $\calH$, the hash table $T$ and an item $(\id, v) \in U^\id \times U^L$ as input and returns an updated hash table $T'$ with $(\id, v)$ added or
    $\perp$ otherwise.
    \item $T \leftarrow \delete(\calH, T, \id)$: A delete algorithm that is given
    the set of hash functions $\calH$, the hash table $T$ and an identifier $\id \in U^\id$ as input and returns an updated hash table $T'$ with $(\id, v)$ removed or
    $\perp$ otherwise.
\end{itemize}}
\end{definition}

As mentioned earlier, we 
consider $\construct$ that enables a hashing
scheme to get all items $X$ that need to be
allocated at once. We choose to focus on this setting as it more closely aligns to the usage
of cuckoo hashing in cryptographic applications
where one party encodes
their entire input using cuckoo hashing.

Next, we move onto the definition of error probabilities for
hashing schemes. We will focus on the
notion of {\em construction error probabilities} that measures
the probability that a set $X$ of identifier-value pairs cannot be constructed into a hash table
according to the public parameters and the sampled set of hash functions
over the randomness of the sampling and construct algorithms.
We emphasize that this definition assumes that the input set
$X$ is chosen independent of the sampled hash functions
(see Section~\ref{sec:robust_def}
for stronger definitions).

\begin{definition}[Construction Error Probability]
A hashing scheme
%$(\sample, \construct, \query)$
for size $n$, identifier universe $U^\id$ and value universe $U^V$ has construction
error probability $\epsilon$ if, for any set of items $X = \{(\id_1,v_1),\ldots,(\id_n,v_n)\} \subseteq U^\id \times U^V$ such that $\id_i \ne \id_j$ for all $i \ne j \in [n]$,
the following holds:
$$
\Pr[\construct(H, X) =\ \perp : H \leftarrow \sample(1^\lambda)] \le \epsilon. 
$$
\end{definition}

In cryptography, the typical requirement for $\epsilon$ would be to be negligible in
the input size $n$. However, in many constructions, cuckoo hashing
is used as a sub-system on a smaller subset of the input size (for example,
subsets of $\log n$ size). For these settings, $\epsilon$ is still required
to be negligible in $n$ even though the cuckoo hashing scheme considers
significantly less than $n$ items. Therefore, the above definition considers generic error probability $\epsilon$ as
there are various settings where the
error probability may need to be much smaller than negligible.

Throughout our paper, we will consider {\em perfect construction algorithms}.
A perfect construction algorithm will always be able to find a
successful allocation for the input set $X$ if at least one such allocation exists.
The main benefit of perfect construction algorithms is that construction
failures only occur if the set of sampled hash functions $\calH$ does not emit a proper
allocation for the input set $X$.
Even with this restriction, we obtain asymptotically optimal
query overhead that match our lower bounds.
We formally define these algorithms below:

\begin{definition}[Perfect Construction Algorithms]
A hashing scheme
for size $n$, identifier universe $U^\id$ and value universe $U^V$ has a perfect construction algorithm $\construct$
if, for any set of items $X = \{(\id_1,v_1),$ $\ldots,$ $(\id_n,v_n)\} \subseteq U^\id \times U^V$ such that $\id_i \ne \id_j$ for all $i \ne j \in [n]$, the following holds:
\[
\Pr\left[\construct(\calH, X) \ne\ \perp : 
\begin{array}{l}\calH \leftarrow \sample(1^\lambda)\\
\exists T \text{ a successful allocation of $X$ according to $\calH$}\\
\end{array}\right] = 1.
\]
\end{definition}

Next, we also
consider query error probability.
Throughout our work, we will
only consider hashing schemes
with zero query error
probability. In other words, if the construction algorithm succeeds, every query will always be correct.

\begin{definition}[Query Error Probability]
A hashing scheme
for size $n$, identifier universe $U^\id$ and value universe $U^V$ has query
error probability $\epsilon_q$ if, for any set of items $X = \{(\id_1,v_1),\ldots,(\id_n,v_n)\} \subseteq U^\id \times U^V$ such that $\id_i \ne \id_j$ for all $i \ne j \in [n]$ and for any query $\id \in U^\id$,
    \[
\Pr\left[\query(\calH, T, \id) \ne v_{\id}: 
\begin{array}{l}\calH \leftarrow \sample(1^{\lambda})\\
T \leftarrow \construct(\calH, X), 
T \ne \perp\\
\end{array}\right] \le \epsilon_q
\]
where $v_\id = v_q$ if $\id = \id_q$ or $v_\id =\ \perp$ otherwise.
\end{definition}

\subsection{Robust Hashing Schemes}
\label{sec:robust_def}

In the prior section, we defined the construction error probability
with respect to input sets $X$ of identifier-value pairs that are
chosen independently of the sampled hash functions. We define
the notion of adversarially robust hashing schemes where an adversary
is given the sampled hash functions $\calH$ and aims to produce
a set $X$ of identifier-value pairs that will fail to allocate.

\begin{definition}[$(Q, \epsilon)$-Robust Hashing Schemes]
A hashing scheme
%$(\sample, \construct, \query)$
for size $n$, security parameter $\lambda$, identifier universe $U^\id$ and value universe $U^V$ is $(Q, \epsilon)$-robust if, for any adversary $\calA$ with running time $O(Q)$,
the following holds:
\[
\Pr\left[\construct(\calH, X) =\ \perp : 
\begin{array}{l}\calH \leftarrow \sample(1^{\lambda})\\
X = \{(\id_1,v_1),\ldots,(\id_n,v_n)\} \leftarrow \calA(\calH)\\
\id_i \ne \id_j, \forall i \ne j \in [n]\\
\end{array}\right] \le \epsilon.
\]
\end{definition}

Once again, we define robustness in a more fine-grained manner
for adversaries running in expected time in $O(Q)$ and arbitrary probabilities
$\epsilon$. Typically, we would use $Q = \poly(n)$ to consider efficient adversaries
and $\epsilon$ to be negligible in $n$. As mentioned earlier, cuckoo hashing
may be used as a sub-system for smaller inputs where we have to consider adversaries
with running time larger than polynomial in the cuckoo hashing size and
probabilities smaller than negligible in the cuckoo hashing size.

Finally, we define
strongly robust to consider all polynomial time adversaries.

\begin{definition}[Strongly Robust Hashing Schemes]
A hashing scheme
%$(\sample,$ $\construct, \query)$
is $(\lambda, \epsilon)$-strongly robust if, for any polynomial $t(n, \lambda)$,
it is $(t, \epsilon)$-robust.
\ignore{
the following holds:
\[
\Pr\left[\construct(\calH, X) =\ \perp : 
\begin{array}{l}\calH \leftarrow \sample(1^{\lambda+n})\\
X = \{(\id_1,v_1),\ldots,(\id_n,v_n)\} \leftarrow \calA(1^\lambda, \calH)\\
\id_i \ne \id_j, \forall i \ne j \in [n]\\
\end{array}\right] \le \epsilon.
\]}
\end{definition}

\section{Cuckoo Hashing}
\label{sec:cuckoo_hashing}
In this section, we re-visit cuckoo hashing. We will consider a generic version of
cuckoo hashing that considers arbitrary
numbers of hash functions $k$, number of entries $b$, entry sizes $\ell$
and overflow stash sizes $s$.
We will exclusively consider the variant of cuckoo hashing
with $k$ disjoint sub-tables of size $b/k$ such that each item is assigned
to one entry in each sub-table according to the $k$ hash functions.

\subsection{Description}
\label{sec:tables}

Cuckoo hashing aims to allocate a set $X$ of $n$ identifier-value 
pairs into $k \ge 2$
disjoint
sub-tables $T_1,\ldots,T_k$ with $b/k$ entries
in each table.
For convenience, we will assume that $k$ divides $b$ evenly.
Each entry is able to store at most
$\ell \ge 1$ items.
The hash table may also consist
of an overflow stash that may be able to store at most
$s \ge 0$ items that were not allocated into
any of the $k$ tables.

Each identifier-value pair is mapped to a random
entry in each of the $k$ tables using $k$ random hash functions,
$(H_1,\ldots,H_k)$ such that $H_i: \{0,1\}^* \rightarrow [b/k]$.
For convenience, we can use a single hash function $H$
that can simulate $k$ hash functions
by setting $H_i(\cdot) = H(i \mid\mid \cdot)$.
Therefore, the $\sample$ algorithm for cuckoo hashing
simply samples a hash function $H$.
See Section~\ref{sec:random_functions} for
various choices of these
random hash functions.
Any identifier-value pair $(\id, v)$ is mapped to the $H_1(\id)$-th entry of $T_1$, the $H_2(\id)$-th entry of $T_2$ and so forth. The pair $(\id, v)$
is guaranteed to be stored in any of the $k$ entries specified
by $H_1,\ldots,H_k$ or the overflow stash.
The $\query$ algorithm checks all possible locations
for the queried item including
the $k$ entries specified by $H_1,\ldots,H_k$
as well as the overflow stash.
So, the query overhead is exactly $k\ell + s$.
Furthermore, assuming
the construction is successful,
the query algorithm will never fail to provide the correct answer (i.e., the query error will always be $0$).
For the $\construct$ algorithm,
there are several options that we will outline later in Section~\ref{sec:algs} once 
we have defined the necessary graph terminology.

We also define the storage overhead describing the size
of the resulting table. For our parameters, the
storage overhead is $b\ell + s$ for the $b$ entries and the stash.

\begin{definition}
The cuckoo hashing scheme $\CH(k, b, \ell, s)$ refers to the algorithm using $k$ hash functions, $b$ entries across $k$ disjoint tables storing at most $\ell$ items and an overflow
stash storing at most $s$ items
with $\sample$ and $\query$ as described
above and any perfect construction
algorithm $\construct$ from Section~\ref{sec:algs}.
The query overhead is $k\ell + s$ and the storage overhead
is $b\ell + s$.
\end{definition}

\noindent{\bf Discussion about Static Tables.}
Dynamic variants of cuckoo hashing enable inserting items into a table. We chose to study
the static variant as insertions are not used
in cryptographic applications.This is due to the fact that
insertions are highly dependent on whether certain entries are populated or empty that is detrimental for privacy
\iffull
(see Appendix~\ref{app:discuss} for further discussion).
\else
(see full version for further discussion).
\fi

\smallskip\noindent{\bf Discussion about Non-Adaptive Querying.}
In our abstraction, we define the query overhead
to be $k\ell + s$ that is total number of possible locations
for any item.
This is due to the fact that we assume non-adaptive querying where all possible locations for
the queried item in one round.
Instead, one could consider an adaptive approach
where locations that are more likely
to contain the queried item are retrieved first (such as non-stash locations). If the item is not found, the query algorithm can proceed with the remaining locations (such as stash locations). This approach has two downsides. First, the querier and cuckoo table storage provider are different parties. Therefore, the
above adaptive query would require multiple roundtrips.
Secondly, the querying algorithm reveals whether certain
entries are populated or empty that is detrimental for privacy (similar to insertions). For these
reasons, to our knowledge, all usages of cuckoo hashing in cryptography
rely on non-adaptive querying with overhead
$k\ell + s$.
See
\iffull
Appendix~\ref{app:discuss}
\else
the full version
\fi
for more details about the
problems with adaptive queries.

\subsection{Cuckoo Bipartite Graphs}
\label{sec:graph}

To be able to analyze cuckoo hashing, we define the notion
of cuckoo bipartite graphs to accurately
model the behavior of cuckoo hashing.
The left vertex set will consist
of the $n$ items that should be allocated.
The right vertex set represents all potential slots
that an item can be stored. If a cuckoo hashing scheme
has $b$ entries each storing at most $\ell$ items and an stash
storing at most $s$ items, there will be $b\ell + s$ total
right vertices.
An edge between an item vertex $v$ and an entry vertex $v'$
means that the item may be allocated to the corresponding entry.
As each item can be stored in at most $k$ entries
or one of the $s$ slots in the stash, each left vertex
will have degree exactly $k\ell + s$.
Each of the $k$ hash functions will be a random function,
so
the above description can be modelled as drawing
randomly from a distribution of bipartite
graphs that we define as follows.

\begin{definition}
Let $\calG(n, k, b, \ell, s)$ denote
the distribution of random bipartite cuckoo graph generated in
the following way:
\begin{itemize}
    \item The left vertex set contains $n$ vertices
    representing the $n$ items to be allocated.
    \item The right vertex set contains $b\ell + s$ vertices that are partitioned into $b+1$ sub-groups in the following way.
    The first $b$ sub-groups contain $\ell$ vertices while
    the last sub-group $S$ contains $s$ vertices corresponding to $b$ entries and the stash respectively.
    We further group together the first $b$ sub-groups (i.e., $b$ entries) as follows. The group $B_1$ consists of the first $b/k$ sub-groups, the group $B_2$ consists of the second
    $b/k$ sub-groups and so forth
    to obtain $k$ groups $B_1,\ldots,B_k$ that each
    correspond to a disjoint table
    of $b/k$ entries each.
    \item Each left vertex is connected to $k\ell + s$ vertices by choosing
    a uniformly random sub-group from each group $B_1,\ldots,B_k$ and adding an edge to all $\ell$ vertices in each of the $k$ chosen sub-groups corresponding to picking a random entry of size $\ell$ from each of the $k$ tables. Finally, each vertex is assigned to all $s$ vertices in $S$ corresponding to the stash.
\end{itemize}
\end{definition}

The major benefit of modelling cuckoo hashing in this manner is that we can
directly map item allocations to matchings in the
bipartite graph. In particular, we can choose an edge between an item and entry slot if and only if that item was allocated
into that entry's slot. Therefore, we can see that any
successful allocation directly corresponds to a {\em left perfect matching} meaning that there exists a set of edges where
each left vertex is adjacent to exactly one edge and each
right vertex is adjacent to at most one edge.
Throughout our paper, our proofs will consist of analyzing the probability
of the existence of left perfect matchings
for various parameter settings
for random graphs drawn from the distribution $\calG(n, k, b, \ell, s)$.

As a note, the entire bipartite graph may be quite large but
does not need to be represented explicitly. In particular,
the graph can be fully re-created using only the parameters $(n, k, b, \ell, s)$, the hash functions
$H_1,\ldots,H_k$ and the set of items to be allocated.
Additionally, we note that any allocation of $n$ items
can be stored using $O(n)$ storage of the corresponding
edges.

\subsection{Perfect Construction Algorithms}
\label{sec:algs}

Finally, we present perfect construction algorithms.
We note that one can rephrase a construction
algorithm as finding a perfect left matching. This amounts
to finding an alternating path from each node to a free right vertex (i.e., empty entry). To do this, we could perform breadth first search (BFS)
starting from each of the $n$ left vertices that guarantees a perfect
construction algorithm. One can also use the more popular random walk
algorithm. However, random walks are not guaranteed to terminate. To
make it a perfect construction algorithm, one can bound the random
walk to $O(n)$ length before running BFS.
Finally, one can also use the local search
allocation algorithm of~\cite{khosla2013balls} that runs
in $O(nk)$ time with high probability.
In
\iffull
Appendix~\ref{sec:running_times},
\else
the full version~\cite{cryptoeprint:2022/1455},
\fi
we describe the
construction algorithms in detail
and analyze the running times.
As most prior works consider constant $k$, we modify their proofs to obtain bounds
for super-constant values of $k$.

\section{Cuckoo Hashing with Negligible Failure}
\label{sec:suff}
We will systematically study
cuckoo hashing across all four parameters
of the number of hash functions $k$, the number
of entries $b$, entry size $\ell$ and
stash size $s$. As our major contribution,
we show that using large $k$ with $k$ disjoint sub-tables
obtains the smallest query overhead for any failure probability $\epsilon$.
We present our construction with large $k$ below:

\begin{theorem}\label{thm:k}
\label{thm:main}
If $H$ is a $(nk)$-wise independent hash function, then
the cuckoo hashing scheme $\CH(k, b, \ell, s)$ has construction failure probability at most $\epsilon$ when
$k = O(1 + \sqrt{\log(1/\epsilon)/\log n})$, $\ell = 1$,
    $s = 0$ and $b = O(n)$.
    The query overhead is $O(1 + \sqrt{\log(1/\epsilon)/\log n})$
    and storage overhead is $O(n)$.
\end{theorem}

We can compare with the best query overhead achievable by prior schemes. Cuckoo hashing with a stash requires
$s = O(\log(1/\epsilon)/\log n)$~\cite{aumuller2014explicit}
and cuckoo hashing with larger entries requires
$\ell = O(\log(1/\epsilon/\log n)$~\cite{minaud2020note}. 
The resulting query overheads of these
instantiations is $O(\log(1/\epsilon)/\log n)$ that is quadratically larger
than the our construction for any $\epsilon \le 1/n$. This includes
negligible failure probability $\epsilon = \negl(n)$ that is typically required in cryptography.
For convenience, we will consider $\epsilon \le 1/n$ for the remainder
of this section so that we can write $k = O(1 + \sqrt{\log(1/\epsilon)/\log n}) = O(\sqrt{\log(1/\epsilon)/\log n})$.

We will also show that all these parameter dependencies are asymptotically
optimal by proving a matching lower bound in Section~\ref{sec:negl_lb}.
In other words, we show that the gap in efficiency is inherent and
cuckoo hashing with more hash functions and disjoint sub-tables is the most
efficient approach.

\smallskip\noindent{\bf Different Parameter Sets.}
In our construction, we considered
extreme parameter regime of large $k$.
One could also consider
parameters that aim to balance between various parameter choices using our techniques. However, it turns out
that the best choice remains using large values of $k$.
We refer to Section~\ref{sec:negl_lb} for further discussions
using the lower bound.

\smallskip\noindent{\bf Different Values for $b$.}
Throughout our work, we considered fixed values of $b = O(n/\ell)$.
We did this to ensure that we restricted to constructions
with $O(n)$ storage overhead.
For completeness, we present tight bounds for parameters with large $b$ in
\iffull
Appendix~\ref{sec:num_entries}.
\else
the full version~\cite{cryptoeprint:2022/1455}.
\fi
We show one must have $b = \Omega(1/\epsilon)$ when considering
parameters with small number of hash functions $k$, entry sizes $\ell$
and stashes $s$. For small $\epsilon < 1/n$, this would result
in super-linear storage overhead.

\smallskip\noindent{\bf Choice of Random Hash Function.} In our above results,
we assumed that $H$ is a $(nk)$-wise independent
hash function. One could,
instead, plug in a $(nk, \epsilon_H)$-wise
independent hash function and obtain 
similar results with
construction failure probability increased by an
additive $\epsilon_H$ factor.
One could also use PRFs or random oracles for $H$
requiring stronger assumptions.

\subsection{Technical Lemmas}

We start with a technical lemma
that relates the existence of an allocation
of $n$ items to perfect left matchings
in bipartite graphs. From there,
we can utilize Hall's Theorem~\cite{hall1987representatives}
to get a very simple characterization
of when an allocation for any $n$
items via cuckoo hashing exists.
We abstract out these lemmas
as we will re-use them when constructing
robust cuckoo hashing in Section~\ref{sec:robust_constructions}.
We note similar analytical tools were used
in the past (such as~\cite{fotakis2005space,aumuller2014explicit}).

For any subset of left nodes $X$, we denote the neighborhood $N(X)$
as the subset of all right nodes that are directly connected
a left node in the set $X$.
Using neighborhoods, we get the following characterization:

\begin{lemma}
\label{lem:hall}
Consider any cuckoo hashing scheme $\CH(k, b, \ell, s)$ with a perfect insertion algorithm
where $H$ is a $(qk)$-wise independent hash function. Then, for any set $S$ of $q$ items,
the construction failure probability
is equal to the probability
that there exists a subset of left vertices $X$ such that $|N(X)| < |X|$
for 
a random graph drawn from the distribution $\calG(q, k, b, \ell, s)$.
\end{lemma}
\begin{proof}
As we consider cuckoo hashing schemes with perfect insertion
algorithms, we know that if there exists a proper
allocation that successfully allocates all
$q$ items, then the construction will be successful.
Therefore, the cuckoo hashing scheme fails to insert
if and only if no allocation fitting all the $q$ items exists.

We can directly map a successful allocation of the $q$
items to a perfect left matching in the corresponding cuckoo graph. Consider any cuckoo hashing
scheme given $q$ items to construct after
fixing the hash function. Then, there exists a corresponding
bipartite graph $G$ in the support of the distribution
$\calG(q, k, b, \ell, s)$. There exists a correct allocation
if and only if each of the $q$ items is assigned uniquely to
a bin or stash location. In other words, a correct allocation
exists if and only if there exists a left perfect matching in $G$
where the allocation of an item to a bin corresponds to an edge in the matching.
By Hall's Theorem~\cite{hall1987representatives}, no left perfect matching exists in a bipartite
graph if and only if there exists some subset $X$ of left vertices
such that its neighbor vertex set is strictly smaller than $X$,
that is, $|N(X)| < |X|$.
\end{proof}

\begin{lemma}
\label{lem:core}
Let $q, k, b, \ell \ge 1$ and $s \ge 0$ and consider the
distribution of random bipartite cuckoo graphs $\calG(q, k, b, \ell, s)$.
The probability that there exists a subset $X$ of left vertices of size $t$
such that it has less than $t$ neighbors, $|N(X)| < t$,
for any $k\ell + s + 1 \le t \le \min\{q, b/2\}$
is at most
$$
\Pr[\exists X: |X| = t, |N(X)| < t] \le {q \choose t} \cdot {b \choose \floor{(t - s - 1)/\ell}} \cdot \left(\frac{2(t- s - 1)}{b\ell}\right)^{k \cdot t}.
$$
For $t \le k\ell + s$, the above probability is $0$.
\end{lemma}
\begin{proof}
Fix any vertex set $X$ of size $t$. For
the case of $t \le k\ell + s$, it is impossible to find a subset $X$ since every left vertex has degree $k\ell + s$.
Consider $t > k\ell + s$.
We know that $X$ immediately has $s$
neighbors in the stash vertices.
Therefore, $X$ must have
at most $t - s - 1$ neighbors outside
of the stash vertices. So,
$X$ must connect to at most
$\floor{(t - s- 1)/\ell}$ bins that consist
of $\ell$ vertices each.
Suppose that these bins are chosen
such that $a_1$ come from the first table,
$a_2$ come from the second table and so forth
such that $a_1+\ldots+a_k \le \floor{(t-s-1)/\ell}$.
We note that the total number of ways
to choose these bins is at most
${b \choose \floor{(t - s - 1)/\ell}}$.
For any vertex $x \in X$, the probability that the $k$ random hash functions pick edges in these bins
is at most
$
(a_1 k/b) \cdots (a_k k / b)
$
as there are $b/k$ bins in each of the tables. Therefore, we get
the probability upper bound of
$$
\Pr[|N(X)| < t] \le {b \choose a_1 + \ldots + a_k} \cdot \left(\prod\limits_{i = 1}^k \frac{a_i k}{b}\right)^{t} \le {b \choose \floor{(t-s-1)/\ell}} \cdot \left(\prod\limits_{i = 1}^k \frac{a_i k}{b}\right)^{t}
$$
since $t \le b/2$.
The right side of the equation is maximized when the product $a_1 \cdots a_k$ is maximized. Therefore,
we get an upper bound by setting
each $a_i = \ceil{(t-s-1)/(\ell k)} \le 2(t-s-1)/(\ell k)$. Plugging this in as well as taking a final Union Bound over all ${q \choose t}$ choices of $X$ we get the following upper bound
$$
\Pr[\exists X: |X| = t, |N(X)| < t] \le {q \choose t} \cdot {b \choose \floor{(t - s - 1)/\ell}} \cdot \left(\frac{2(t- s - 1)}{b\ell}\right)^{k \cdot t}
$$
to complete the proof.
\end{proof}

\subsection{Our Construction}
\label{sec:failure_k}

Next, we prove that our
cuckoo hashing construction
with more hash functions (i.e., larger $k$) results
in quadratically smaller
query overhead.
Recall our construction from Theorem~\ref{thm:k}
uses $k = O(\sqrt{\log(1/\epsilon)/\log n})$ hash functions,
$b = O(n)$ entries of size $\ell = 1$ and no stash, $s = 0$.

\begin{proof}[Proof of Theorem~\ref{thm:k}]
To prove this, we will leverage Lemma~\ref{lem:hall}
that provides
a tight characterization between a successful insertion
in cuckoo hashing and left perfect matchings in random
bipartite graphs.
We start from
the probability upper bound from
Lemma~\ref{lem:core} and plug in
our values of $k, \ell, b$ and $s$ to get the following
for values of $k + 1 \le t \le n$
assuming $b \ge 2n$:
\begin{align*}
\Pr[\exists X: |X| = t, |N(X)| < |X|] &\le {n \choose t} \cdot {b \choose t - 1} \cdot \left(\frac{2(t-1)}{b}\right)^{k \cdot t}\\
&\le \left(\frac{en}{t}\right)^t \cdot \left(\frac{eb}{t - 1}\right)^{t-1} \cdot \left(\frac{2(t-1)}{b} \right)^{k \cdot t}\\
&\le \left(\frac{eb}{t-1}\right)^{2(t-1)} \cdot \left(\frac{eb}{t - 1}\right)^{t-1} \cdot \left(\frac{2(t-1)}{b} \right)^{k \cdot t}\\
&\le (2e)^{kt} \cdot \left(\frac{t-1}{b} \right)^{(k-3)t}.
\end{align*}
Note that we used Stirling's approximation that ${x \choose y} \le (ex/y)^y$ in the second inequality.
For the second and third inequality, we use
that $n \le b$
and $t \ge k \ge 3$.
Note, if we set 
$b \ge (2e)^{5} \cdot n$ appropriately and assume that $k \ge 4$, we get
the following inequality:
$$
\Pr[\exists X: |X| = t, |N(X)| < |X|] \le \left(\frac{t-1}{2n} \right)^{(k-3)(t-1)}.
$$

We break the analysis into two parts
depending on the value of $t$. We start
with the case that for the smaller range $k+1 \le t \le n^{0.75}$.
As a result, we can upper bound
the probability by
$(1/n^{0.25})^{(k-3)(k-1)}$.
For the other case, assume that $n^{0.75} < t \le n$.
Therefore, we
can upper bound the probability by
$(1/2)^{(k-3) \cdot (n^{0.75})}
\le (1/n)^{(k-3)(n^{0.75}/\log n)}$.
Finally, we know that
$$\max\{(1/n^{0.25})^{(k-3)(k-1)}, (1/n)^{(k-3)(n^{0.75}/\log n)}\} \le (1/n^{0.25})^{(k-3)(k - 1)}$$
for sufficiently large $n$.
Applying a Union Bound for all values of $k + 1 \le t \le n$,
we get that
$$
\Pr[\exists X: |N(X)| < |X|] \le (n - k) \cdot (1/n^{0.25})^{(k-3)(k-1)}
%\le \left(1/n^{0.25}\right)^{(k-3)(k-1) - 4}
= \left(1/n\right)^{\Theta(k^2)}.
$$
We solve the following inequality
to get a lower bound on $k$ based on $\epsilon$
$$
(1/n)^{\Theta(k^2)} \le \epsilon \implies k = O( \sqrt{\log(1/\epsilon)/\log n}).
$$
As $k$ must be at least $4$, we get that $k = O(1 + \sqrt{\log(1/\epsilon)/\log n})$.
\end{proof}

\noindent{\bf Necessity of $k$ Disjoint Tables.} Prior work~\cite{fotakis2005space}
aimed to analyze the failure probabilities for cuckoo hashing with arbitrary $k$ hash functions when $b = O(n)$, $\ell = 1$ and $s = 0$. Rephrasing
their results, the prior result required $k = O(\log(1/\epsilon)/\log n)$ that is quadratically higher
than our result. The core difference between the two results
is that our work analyzes the setting where there are $k$
disjoint sub-tables while the
prior result~\cite{fotakis2005space} considered a single shared table. With $k$ disjoint sub-tables, we guarantee that any set of at most $k$ left vertices will have $k$ distinct neighbors.
Therefore, our analysis only needs
to consider left vertex sets of larger size. For the setting where each of the $k$ hash functions may choose
any of the $b$ entries in a single shared table, we note a result similar to ours is impossible.
Consider the setting of left vertex sets of size 2, $|X| = 2$. The probability
that all $2k$ hash function evaluations resulting in the same
entry is already $(1/b)^{2k}$. If $b = \Theta(n)$, this immediately
implies that $k = \Omega(\log(1/\epsilon)/\log n)$ for failure probability $\epsilon$. By
avoiding this case using disjoint tables, we
are able to obtain the same failure probability
with a quadratically smaller number of hash functions.

\subsection{Lower Bounds}
\label{sec:negl_lb}

Next, we prove
lower bounds on the best possible
parameters obtainable
in cuckoo hashing.
In particular, we will show
that the chosen parameters in Theorem~\ref{thm:k} are asymptotically
optimal for
failure probabilities $\epsilon$.

Our lower bounds do
not make any assumptions on
the construction algorithm used.
In other words, the results apply regardless of the construction algorithm (such as whether they are perfect or whether they are efficient). Additionally, we only make the assumption that the underlying hash
function $H$ is a $(k\ell + s + 1)$-wise
independent hash function to mimic
standard cuckoo hashing constructions.

At a high level, we will present
a simple attack relying on the insight
that cuckoo hashing failure is hard for small sets.
By the structure of our cuckoo hashing scheme
using $k$ disjoint sub-tables and a stash,
we know that any set of $k\ell + s$
will be allocated correctly.
Our goal is to simply pick a random
set of $k\ell + s + 1$ items and
lower bound the probability
that all these items will hash into
the exact same $k$ entries. In this case,
the construction algorithm
would fail.

\begin{theorem}
\label{lem:err_lb}
Let $k \ge 1$, $\ell \ge 1$, $1 \le b \le n^{O(1)}$, $s \ge 0$ such that $k\ell + s + 1 \le n$.
The failure probability of $\CH(k, b, \ell, s)$ cuckoo hashing scheme where $H$ is a $(k\ell + s + 1)$-wise independent hash function
satisfies the following:
\begin{align}
%\label{eq:attack}
k^2 \ell + ks = \Omega(\log(1/\epsilon)/\log(n)).
\end{align}
\end{theorem}
\begin{proof}
We know that the first item will be successfully inserted.
Consider the $k$ distinct locations that
were chosen for the first item denoted by $S$. Suppose that another $k\ell + s$
different items were also assigned to the same $k$ locations
or a subset of the $k$ locations. In this case, there
are $k\ell + s +1$ items that must be assigned to $k\ell + s$ locations 
in the $k$ entries and the stash that is impossible and will result in an insertion failure regardless of the choice of the insertion algorithm.
To obtain our lower bound, we simply lower bound this probability.
Consider the other $k\ell + s$ to be
inserted. Each of these $k\ell + s$ items will pick $k$ locations
uniformly at random from each of the $k$ disjoint tables. Therefore, the probability
that the $k$ choices will be a subset of $S$ is $(k/b)^k$.
As all choices are independent, the probability that
this is true for all $k\ell$ items is $(k/b)^{k(k\ell + s)}$. Therefore,
the probability of an insertion failure is at least
$\epsilon \ge (k/b)^{k^2\ell + ks}$.
Applying logs to both sides gets that $k^2\ell + ks \ge \log(1/\epsilon)/\log(b/k)$.
Using
the fact that $b \le n^{O(1)}$ and $k \ge 1$, we get that $\log(b/k) = O(\log n)$.
%As $\ell k \le n^{4/5}$, we get that
%$\log(b/k) = \Omega(\log n)$.
Therefore, we get the inequality $k^2 \ell + ks = \Omega(\log (1/\epsilon)/\log n)$ completing the proof.
\end{proof}

\begin{theorem}
\label{thm:general_lb}
Let $1 \le k \le n^{2/5}$ and $1 \le \ell \le n^{2/5}$.
Suppose that the cuckoo hashing scheme $\CH(k, b, \ell, s)$ where $H$ is a $(k\ell+s+1)$-wise independent
hash function. Then, the following are true:
\begin{itemize}
    \item If $\ell = O(1), b = n^{O(1)}$ and $s = O(1)$, then $k = \Omega(\sqrt{\log(1/\epsilon)/\log n})$.
    \item If $k = O(1), b = n^{O(1)}$ and $s = O(1)$, then $\ell = \Omega(\log(1/\epsilon)/\log n)$.
    \item If $k = O(1), b = n^{O(1)}$ and $\ell = O(1)$, then $s = \Omega(\log(1/\epsilon)/\log n)$.
    %\item If $k = O(1), \ell = O(1)$ and $s = O(1)$, then
    %$b = \Omega(1/\epsilon)$.
\end{itemize}
\end{theorem}
\begin{proof}
Plug in the values for each parameter
regime into Theorem~\ref{lem:err_lb}.
\end{proof}

We note that the above corollary shows that
our construction in Theorem~\ref{thm:k} is asymptotically
optimal. Furthermore, we also show that the constructions
of large stashes~\cite{aumuller2014explicit} and
large entries~\cite{minaud2020note} are also tight.

\smallskip\noindent{\bf Balancing Parameters.}
For our constructions, we only consider the extreme regime
of large $k$ while the other parameters $s$ and $\ell$ remain
very small. Instead, we could consider trying to balance the
parameters to obtain more efficient constructions. Using Theorem~\ref{lem:err_lb}, we can see why the choice of large $k$
is the most efficient approach. Recall that the query overhead
is $O(k\ell + s)$. However, it must be that
$k^2\ell + ks = \Omega(\log(1/\epsilon)/\log n)$. In other words,
we want to minimize $k\ell + s$ while satisfying the above condition.
It is not hard to see that the optimal approach
is to set $k = O(\sqrt{\log(1/\epsilon)/\log n})$
as we do in our constructions.

\smallskip\noindent{\bf Implications to Other Primitives.} As
our lower bounds do not make any assumptions on
the construction algorithm,
our results can apply to other hashing
schemes other than cuckoo hashing.
For example, we can consider multiple-choice allocation schemes~\cite{azar1994balanced,richa2001power} where items
are allocated to multiple entries and placed
into the entry with the current smallest load.
In general, one can apply our lower
bounds to any scheme that limits the
candidate entries for any item to at most $k$
locations and can set an upper limit on
the number of items per entry to $\ell$.

Our lower bounds
apply to other primitives that heavily
rely on cuckoo hashing techniques such as cuckoo filters~\cite{fan2014cuckoo}
and oblivious key-value stores~\cite{garimella2021oblivious} that
encode data using cuckoo hashing (see
\iffull
Appendix~\ref{app:encoding}
\else
the full version~\cite{cryptoeprint:2022/1455}
\fi
for more details).

\smallskip\noindent{\bf Assumptions in Our
Lower Bounds.}
In the above theorem, we made the assumptions that $k \le n^{2/5}$ and $\ell \le n^{2/5}$.
We do not believe these limit the applicability of our lower
bounds as, otherwise, the query overhead of cuckoo hashing
will be too large. As query overhead is $O(k\ell + s)$, either condition being false immediately implies
$\Omega(n^{2/5})$ query overhead that is impractical.

\smallskip\noindent{\bf One or Two Hash Functions.}
These lower bounds match constructions for cuckoo hashing
with large stashes~\cite{aumuller2014explicit} and entries~\cite{minaud2020note} with $k = 2$. Our lower
bound does not preclude obtaining the same results with $k = 1$.
However, it turns out that $k \ge 2$ is necessary as one can
use analysis from ``balls-into-bins'' analysis to show that $k = 1$
is impossible. For completeness, we include impossibility results for $k = 1$ in
\iffull
Appendix~\ref{sec:single_hash_function}.
\else
the full version of this paper~\cite{cryptoeprint:2022/1455}.
\fi

\smallskip\noindent{\bf Comparison with
Prior Lower Bounds.}
We note that prior work~\cite{minaud2020note}
also proved lower bounds for constant values of $s$ and $\ell$ and fixed $k = 2$. In particular,
for $k = 2$, constant entry size
$\ell \ge 1$, constant stash size $s \ge 0$, and $b = O(n/\ell)$, they proved
that $\epsilon = \Omega(n^{-s-\ell})$.
Our lower bounds improve upon this as we can
consider arbitrary $s$ and $\ell$.

\section{Robust Cuckoo Hashing}
\label{sec:adv}
In this section, we study robust cuckoo hashing where the input set can be chosen by
an adversary that is also given input
to the hash function $H$.
To model this, we consider the adversary having access
to an oracle $\calO$ for hash evaluations.
We will present variants of cuckoo hashing
that can still guarantee smaller construction
failures even when the input set is chosen
adversarially by efficient adversaries
with knowledge of the randomness (that is, the hash function $H$).
In particular, we study $(\lambda,\epsilon)$-strong robustness
where $\poly(\lambda)$ adversaries cannot find a failing input set
except with probability $\epsilon$.

\subsection{Robustness Constructions}
\label{sec:robust_constructions}

Recall
that in the prior section, we required
that $k = O(\sqrt{\log(1/\epsilon)/\log n})$ to
obtain $\epsilon$ construction failure probabilities. We show that increasing the number of hash functions
by a small amount
suffices
to obtain robustness.
To analyze robust cuckoo hashing, we will consider
a hash oracle $\calO$ that may be queried by the adversary.
We consider an adversary with running time $Q$ that may
query at most $Q$ hash evaluations. Afterwards, we analyze
the probability there exists a subset of $n$ items
amongst the $Q$ queried items that would incur a
failure under the hash functions.
We will also show our choice of $k$ is optimal in Section~\ref{sec:robust_lb}.

\begin{lemma}
\label{lem:robust_fine}
\label{lem:robust_k}
For any $0 < \epsilon < 1$,
let $k = O(\log(Q/\epsilon))$, $s = 0$, $\ell = 1$, $b = \alpha n$ for some constant $\alpha \ge 1$.
If $H$ is a random hash function and for any $Q \ge n$,
the cuckoo hashing scheme $\CH(k, b, \ell, s)$ is $(Q,\epsilon)$-robust.
\end{lemma}
\begin{proof}
We will consider hash oracle $\calO$ that
returns $\calO(x) = (H_1(x), \ldots, H_k(x))$ on input $x$.
That is, a single oracle query
will return the outputs of all
$k$ hash functions.
Without loss of generality, we will assume
that if the adversary returns a set $S$ of
$n$ items, then the adversary has
computed $H_1(s), \ldots, H_k(s)$ for
all $s \in S$.
This only
increases the number
of queries of the adversary by
at most $n \le Q$ for a total
of at most $2Q$ hash queries.
Let $U$ be the set of all items
that the adversary has queried
to the hash functions.
That is, if $u \in U$, then
the adversary knows
the values $\calO(u) = (H_1(u),\ldots,H_k(u))$.
We know that $|U| \le 2Q$.

To show that the scheme
is robust, we will show that a random
graph drawn from the distribution
$\calG(|U|, k, b, \ell, s)$ does not contain
any set of left vertices $X$ such
that $|X| \le n$ and $|N(X)| < |X|$.
By proving no such set of left vertices $X$ exists,
it will be impossible
for the adversary to identify
any input set that would
cause the cuckoo hashing scheme $\CH(k, b,\ell,s)$ to fail.

By applying Lemma~\ref{lem:core} with our parameters of $s = 0$, $\ell = 1$ and $b = \Theta(n)$, we get the following
for probability upper bounds for the values of $k + 1 \le t \le n$:
\begin{align*}
\Pr[\exists X: |X| = t, |N(X)| < |X|] &\le {|U| \choose t} \cdot {b \choose t - 1} \cdot \left(\frac{2(t-1)}{b}\right)^{k \cdot t}\\
&\le \left(\frac{2eQ}{t}\right)^t \cdot \left(\frac{eb}{t - 1}\right)^{t-1} \cdot \left(\frac{2(t-1)}{b} \right)^{k \cdot t}\\
&\le (2Q)^{t} \cdot (2e)^{k \cdot } \cdot \left(\frac{t-1}{b} \right)^{(k-3)t}\\
&\le (2Q)^{t} \cdot \left(\frac{t-1}{2n}\right)^{(k-3)t}\\
&\le \left(\frac{t^{k-3} \cdot (2Q)}{(2n)^{k-3}} \right)^t.
\end{align*}
Since $k + 1 \le t \le n$, we can upper bound
the above probability by
$$
\Pr[\exists X: |X| = t, |N(X)| < |X|] \le \left(\frac{2Q}{2^{k-3}}\right)^{k+1}.
$$
As this needs to be at most $\epsilon$,
we can derive the following
inequalities
\begin{align*}
\left(\frac{2Q}{2^{k-3}}\right)^{k+1} \le \epsilon \implies (k+1)(k-3-\log(2Q)) \ge \log(1/\epsilon).
\end{align*}
If we set
$k = O(\log(Q/\epsilon))$,
we get the desired bound
that completes the proof.
\end{proof}

Next, we show that this may be extended to the notion
of strongly robust that applies to any polynomial time adversaries.

\begin{theorem}
\label{thm:robust_k}
For security parameter $\lambda$ and error $0 < \epsilon < 1$, let $k = O(f(\lambda) + \log(1/\epsilon))$ for some function $f(\lambda) = \omega(\log \lambda)$, $s = 0$, $\ell = 1$ and $b = \alpha n$ for some constant $\alpha \ge 1$. If $H$ is a random hash function, then
the cuckoo hashing scheme $\CH(k, b, \ell, s)$
 is $(\lambda, \epsilon)$-strongly robust.
\end{theorem}
\begin{proof}
The adversary
runs in polynomial time and, more importantly, makes at most $\poly(\lambda)$
queries to the hash oracle.
We fix $Q = 2^{\omega(\log \lambda)} = \lambda^{\omega(1)}$ to be any
function super-polynomial in $\lambda$ and, thus, $Q$ is larger
than the running time of any $\poly(\lambda)$ time algorithm.
By applying
Lemma~\ref{lem:robust_k}
with
$Q = 2^{\omega(\log \lambda)}$ to obtain robustness
except with probability $\epsilon$,
we get that $k = O(\log(Q/\epsilon)) = \omega(\log \lambda) + O(\log(1/\epsilon))$ suffices to complete the proof.
\end{proof}

Finally, we can apply our above theorem
for standard values of $\lambda = n$ and
$\epsilon = n^{f(n)}$ for any $f(n) = \omega(1)$ that is negligible
in $n$ to get the following corollary.

\begin{corollary}
\label{cor:robust}
Let $\epsilon = n^{f(n)}$ for any function $f(n) = \omega(1)$.
Let $k = O(f(n) \log n)$, $s = 0$, $\ell = 1$ and $b = \alpha n$ for some constant $\alpha \ge 1$.
Then,
the cuckoo hashing scheme $\CH(k, b, \ell, s)$ is $(n, \negl(n))$-strongly robust.
\end{corollary}

\noindent{\bf Instantiation of Hash Function.}
In this section, we made
the assumption that the hash function
$H$ is indistinguishable from random
and that $H$
may also be queried by the adversary.
We leave it as an
open problem to consider hash functions
from other assumptions.

\subsection{Lower Bounds for Robustness}
\label{sec:robust_lb}

In this section, we prove lower bounds
for the required parameters to ensure
that cuckoo hashing is robust. 
We will show that our construction
in Theorem~\ref{thm:robust_k} is
asymptotically optimal in the regime
of sub-linear stash and entry sizes.
First, we will
assume that the stash and entry size are sub-linear,
$s = o(n)$ and $\ell = o(n)$, and that
the number of entries is $b = O(n/\ell)$.
We prove our lower bound with respect to $(\lambda, 1/2)$-robustness.
As we are proving lower bounds,
our result also applies to smaller, more reasonable, failure
probabilities such as negligible failure.

At a high level, our lower bound consists of a simple adversary.
The goal of the adversary is to find a set of $n$ items
that are allocated into the first half of each of the $k$
disjoint sub-tables. By analyzing the probability of finding
such an input set, we obtain a matching lower bound.

\begin{theorem}
\label{thm:lb_robust_k}
Suppose that $k, \ell \ge 1$ and $s \ge 0$ such that $\ell = o(n)$, $s = o(n)$ and
$b = O(n/\ell)$.
If $H$ is a random hash function and the cuckoo hashing scheme $\CH(k, b, \ell, s)$ is $(\lambda, 1/2)$-robust
for $\lambda \ge n$, then it must be that $k = \omega(\log \lambda)$.
\end{theorem}
\begin{proof}
To prove this, we assume a contradiction
that $k = O(\log \lambda)$ and we will show that there
exists a $\poly(\lambda)$ time adversary that outputs
a set $S$ of $n$ items to insert
such that all $n$ items
hash to the same $n/(2\ell)$ bins
except with at most $1/2$ probability.
In this case, we note that the $n/(2\ell)$ bins and the stash can
store at most
$n/(2\ell) \cdot \ell + s = n/2 + o(n) \le 3n/4$ that cannot store all $n$
items in the set $S$ that would complete the proof.

To construct our adversary, we will
leverage that $k = O(\log \lambda)$. Assuming that
$H$ is a random hash function, we know
that for any input $x$,
$$\Pr[H_1(x) \le n/(2\ell k) \land \ldots \land H_k(x) \le n/(2\ell k)] = \left(\frac{n/(2\ell k)}{b/k}\right)^k = \left(\frac{n}{2\ell b}\right)^k.$$
As $b = O(n/\ell)$, there exists some
constant $\alpha > 0$ such that
$b \le \alpha n/(2\ell)$ meaning that
$(n/(2\ell b))^k \le (1/\alpha)^k$.
Suppose that $k \le c \log_\alpha n$ for
some constant $c > 0$ since $k = O(\log n)$ and $\alpha$ is a positive constant,
then we know that
$\Pr[H_1(x) \le n/(2\ell k) \land \ldots \land H_k(x) \le n/(2\ell k)] \le (1/\alpha)^k \le 1/\lambda^c$.
Next, we construct the following
adversary using the above probability
that aims to find a set $S$ of $n$ items
such that all items satisfy the above
property in the following way:

\smallskip\noindent{\bf Adversary $\calA(H_1,\ldots,H_k)$:}
\begin{enumerate}
    \item Let $S \leftarrow \emptyset$.
    \item Let $\cnt \leftarrow 0$.
    \item While $|S| < n$ and $\cnt \le \lambda^{c+2}$:
    \begin{enumerate}
        \item If $H_1(\cnt) \le n/(2\ell k) \land \ldots \land H_k(\cnt) \le n/(2\ell k)$, set $S \leftarrow S \cup \{\cnt\}$.
        \item Set $\cnt \leftarrow \cnt + 1$.
    \end{enumerate}
    \item If $|S| < n$, return $\perp$.
    \item Return $S$ as the $n$ items to insert.
\end{enumerate}
We analyze a slight modification of the adversary that executes all $\lambda^{c+2}$ iterations before returning. Let $X_i = 1$ if and only if the $i$-th iteration
(when $\cnt = i$) succeeds in being placed into $S$. Let $X = X_1 + \ldots + X_{n^{c+2}}$. We know that the adversary
outputs $\perp$ if and only if 
$X < n$.
Note that $\Pr[X_i = 1] = 1/\lambda^c$, so
$\mu = \E[X] = \lambda^{c+2}/\lambda^c = \lambda^2$.
As each $X_i$ is independent
due to the random hash functions and $\lambda \ge n$, we can apply Chernoff's Bound to get that
$
\Pr[X < n] \le \Pr[X < \lambda^2/2] = \Pr[X < \mu/2] \le 2^{-\Theta(\lambda^2)}
$.
In other words, the adversary
outputs the desired set $S$ with
probability at least $1-2^{-\Theta(\lambda^2)} > 1/2$ as required. Note the
adversary is polynomial time as each of the $\lambda^{c+2}$ iteration requires $O(k) = O(\log \lambda)$ time
by assumption. Therefore, the adversary's running time is $O(\lambda^{c+2} \log \lambda)$ that is polynomial in $\lambda$
as $c$ is a positive constant.
\end{proof}

The above shows that if we consider sub-linear $s$ and $t$, then it must be that $k = \omega(\log \lambda)$. We can consider
the contrapositive of the above theorem. Suppose
that $k = O(\log \lambda)$ that is slightly smaller
than the lower bound above. Assuming that $b = O(n/\ell)$ to ensure storage of the hash table remains linear,
this immediately implies that either $s = \Omega(n)$
or $\ell = \Omega(n)$.
In other words, either the stash or entry must be able to store
almost all $n$ inserted items and the resulting query overhead
is $O(n)$.
These parameter sets are essentially trivial 
as it is equivalent to retrieving the entire cuckoo hash table.

\begin{theorem}
\label{cor:sl_robust_lb}
Suppose that $k, \ell \ge 1$ and $s \ge 0$ such that $k = O(\log \lambda)$ and
$b = O(n/\ell)$.
If the cuckoo hashing scheme $\CH(k, b, \ell, s)$ is $(\lambda, 1/2)$-robust for $\lambda \ge n$, then it must be that $s + \ell = \Omega(n)$ with query overhead $\Omega(n)$.
\end{theorem}

\section{Batch Codes}
\label{sec:pbc_pir}
\subsection{Probabilistic Batch Codes}

The notion of batch codes was introduced by
Ishai, Kushilevitz, Ostrovsky and Sahai~\cite{IKO04}.
At a high level, the goal of a batch code
is to distribute a database of $n$ entries into
$m$ buckets such that any subset of $q$ entries
may be retrieved by querying at most $t$
codewords from each of the $m$ buckets.
The size parameter $N$ denotes the total codewords across all $m$ buckets
and we denote the rate of
the batch code by $n/N$.
When constructing batch codes, the goal is to maximize the rate $n/N$ while keeping the number of buckets
$m$ as close to $q$ as possible, ideally $m = O(q)$, and minimizing $t$, ideally $t = 1$.

Leveraging our results in cuckoo hashing, we will present improved constructions for batch codes. In
particular, we will present a probabilistic batch code (PBC)
with quadratically smaller rate compared to prior works (see Figure~\ref{table:pbc_compare}).

\begin{figure*}[tb]
\centering
\small
\begin{tabular}{| l | c | c | c | c | c |}
        \hline
      \makecell{\bf PBC} & \makecell{\bf Size ($N$)} & \makecell{\bf Buckets ($b$)} &  \makecell{\bf Explicit?} & \makecell{\bf Error}
\\ \hline
    \makecell[l]{Subset~\cite{IKO04}} & $O(n)$ & $q^{O(1)}$ & $\checkmark$ & 0  \\ \hline
    \makecell[l]{Expander Graphs~\cite{IKO04}} & $O(n \log n)$ & $O(q)$ & $\times$ & 0  \\ \hline
        \makecell[l]{Balbuena Graphs~\cite{rawat2016batch}} & $O(n)$ & $O(q^3)$ & $\checkmark$ & 0  \\ \hline
    \makecell[l]{Pung~\cite{angel2016unobservable}} & $4.5n$ & $9q$ & $\checkmark$ & $2^{-20*}$ \\ \hline
    \makecell[l]{3-way Cuckoo Hashing~\cite{angel2018pir}} & $3n$ & $1.5q$ & $\checkmark$ & $2^{-40*}$  \\ \hline
    \hline
    
    \makecell[l]{Our Work} & $O(n \cdot \sqrt{\lambda/\log\log n})$ & $O(q)$ & $\checkmark$ & $2^{-\lambda}$ \\ \hline
\end{tabular}
\caption{A comparison table of prior PBC constructions with $t = 1$. Constructions with only experimental evaluations are marked with asterisks(*).}
\label{table:pbc_compare}
\iffull
\else
\vspace*{-5mm}
\fi
\end{figure*}

Probabilistic
batch codes (PBCs) were introduced by Angel, Chen, Laine and Setty~\cite{angel2018pir}.
Unlike batch codes, PBCs are able to err on a subset
of potential queries with the goal of obtaining more efficient
parameters. To date, state-of-the-art PBCs are built from either
directly adapting batch codes with zero error or constructions
whose error probabilities have only been experimentally evaluated.
By adapting our analysis of cuckoo hashing, we are able to construct
a PBC with provable error probabilities that have better parameters
than all prior works. To our knowledge, our batch code either has quadratically better rate or cubically smaller number of buckets than the
best prior construction (including non-explicit ones).
We point readers to Figure~\ref{table:pbc_compare} for more details.

Before we present our constructions, we formally define the notion of PBCs.
Note, we will construct {\em systematic} or {\em replication} batch codes where each
codeword must be one of the $n$ entries in the database.

\begin{definition}[Probabilistic Batch Codes]
A $(n, N, q, m, t)$-systematic PBC consists of the following four efficient algorithms:
\begin{itemize}
    \item $\params \leftarrow \init(1^\lambda)$: The initialization algorithm takes the security parameter and
    outputs parameters.
    \item $(C_1,\ldots,C_m) \leftarrow \encode(\params,\db)$: The encode algorithm
    takes a database $\db$ of $n$ entries as input
    and outputs $m$ buckets such that the total number of codewords
    is at most $N$. Furthermore,
    each $(C_i)_j$ must be one of the $n$ database entries in the set $\{\db_i\}_{i \in [n]}$.
    \item $(S_1,\ldots,S_m) \leftarrow \genschedule(\params,Q)$: The schedule algorithm
    takes as input a query $Q$ of $q$ distinct elements and outputs a schedule of the indices of each bucket to read such that
    each $|S_i| \le t$.
    \item $A \leftarrow \decode(\params,Q, (C_1)_{i \in S_1},\ldots, (C_m)_{i \in S_m})$: The decode algorithm
    takes as input a query $Q$ of $q$ distinct elements in $[n]$
    and the scheduled indices of each code and outputs the
    queried database entries.
\end{itemize}
Furthermore, the PBC has error at most $\epsilon$ if, for
all database $\db$ and queries $Q$ of $q$ distinct elements, the following holds:
\[
\Pr\left[R \ne (\db_i)_{i \in Q} : 
\begin{array}{l}
\params \leftarrow \init(1^\lambda)\\
(C_1,\ldots,C_m) \leftarrow \encode(\params, \db)\\
(S_1,\ldots,S_m) \leftarrow \genschedule(\params, Q)\\
R \leftarrow \decode(\params, Q, (C_1)_{i \in S_1},\ldots, (C_m)_{i \in S_m})\\
\end{array}\right] \le \epsilon.
\]
\end{definition}

As a note, we only consider PBCs whose queries do not contain
duplicate entries (i.e., each query is to a unique entry).
In most practical applications, the querier can handle duplicate entries
in the scheduling algorithm by removing duplicates and then duplicating
them in the decode algorithm.
This assumption is not limiting in most practical applications such as PIR (see Section~\ref{sec:app_pir}).

Next, we show that a cuckoo hashing scheme may
be used to construct PBCs with similar parameters.
We note that a similar reduction was informally shown in~\cite{angel2018pir} previously. 
\iffull
\else
We present the proof in the full version~\cite{cryptoeprint:2022/1455}.
\fi

\begin{lemma}
\label{lem:ch_pbc}
If there exists a cuckoo hashing $\CH(k, b, \ell, s)$ for $q$ items
that has failure probability at most $\epsilon$,
then there exists a $(n, (k + \ceil{s/\ell})n, q, b + \ceil{s/\ell}, \ell)$-systematic PBC
for a universe $U$ of size $n$ with error at most $\epsilon$.
\end{lemma}
\iffull
\begin{proof}
We convert a cuckoo hashing scheme into a PBC in the following manner. The $\init$ algorithm executes the $\sample$ algorithm
of cuckoo hashing to obtain hash functions $H_1,\ldots,H_k$.
The $\encode$ algorithm creates $b + \ceil{s/\ell}$ buckets representing
the $b$ entries and the stash. For the $\ceil{s/\ell}$ stash buckets, all
$n$ database entries are added to each bucket using $\ceil{s/\ell} \cdot n$ codewords. As at most $s$ items needs to be
retrieved from the stash and we can retrieve $\ell$
codewords from each bucket, only $\ceil{s/\ell}$ buckets are needed.
For the remaining $b$ buckets, each of the $n$ items in $U$
are added to $k$ buckets according to $H_1,\ldots,H_k$ using $kn$ codewords meaning $N = (k+\ceil{s/\ell})n$.

In $\genschedule$, the querier allocates
the $q$ elements in $Q$ using cuckoo hashing. As a result, the
querier can determine the correct indices in each of the $b+s$ buckets needed to decode the $q$ database entries.
Finally, we note that the cuckoo hashing algorithm
fails with probability at most $\epsilon$ for any set
of $q$ items. Therefore, this immediately implies
that $\decode$ has error probability at most $\epsilon$.
\end{proof}
\fi

As an immediate consequence of this lemma, one
can immediately construct a PBC with negligible error
using prior cuckoo hashing with large stash results~\cite{kirsch2010more,aumuller2014explicit,minaud2020note}.
For example, one can obtain a
$(n, O(n\lambda), q, O(q), 1)$-PBC
with error $2^{-\lambda}$.
We omit the proof as it was
already known to exist in folklore.

Using the above reduction, we can construct an efficient PBC
with rate $1/O(\sqrt{\lambda/\log \log n})$
from Theorem~\ref{thm:main}
with quadratically better rate than prior
constructions with $b = O(q)$.
We point to
Figure~\ref{table:pbc_compare} for further comparisons.
\iffull
\else
The proof may be found in the full version~\cite{cryptoeprint:2022/1455}.
\fi

\begin{theorem}
\label{thm:pbc}
For all $1 \le q \le n$,
there exists a $(n, N, q, b, \ell)$-systematic
PBC with
at most
$2^{-\lambda}$ error
where $b = O(q)$, $\ell = 1$ and
$N = O(n \cdot \sqrt{\lambda/\log \log n})$.
\end{theorem}
\iffull
\begin{proof}
We break the construction into
two different regimes for values of $q$.
First, we consider larger values
of $q = \Omega(\sqrt{\lambda/\log \log n})$.
For this case, we utilize the cuckoo hashing with
negligible failure of Theorem~\ref{thm:main}
for a set of $q$ items.
In particular,
we use the setting of
large $k$ with $b = O(q)$, $\ell = 1$ and $s = 0$.
We set $\epsilon = 2^{-\lambda}$
to get that
$k = O(\sqrt{\lambda /\log q})$
where $k \le b$.
Note, this requirement comes from the fact
that we use $k$ disjoint tables that is only
possible when $k \le b$.
To guarantee that $k \le b$, we can also
ensure that $k \le q$ as $b \ge q/\ell = q$
since $\ell = 1$.
From this, we derive the following requirement
for $q$:
$$
O(\sqrt{\lambda /\log q}) = k \le q \implies q^2\log q = \Omega(\lambda).
$$
Plugging in values $q = \Omega(\sqrt{\lambda/\log \log n})$ ensures that the inequality is true.
By applying Lemma~\ref{lem:ch_pbc}, we get the PBC with the desired parameters for
the setting of $q = \Omega(\sqrt{\lambda/\log \log n})$.

For the case of smaller $q = O(\sqrt{\lambda/\log\log n})$, we can use a simple batch code of
$b = q$ bins where each bin stores all $q$ items.
Clearly, this is a batch code with zero error
while
still satisfying the desired parameters to complete the proof.
\end{proof}
\fi

\subsection{Robust Probabilistic Batch Codes}
We introduce the notion of adversarially robust PBCs
that lies in between batch codes with zero error and PBCs with
negligble error. Robust PBCs guarantee that even, if there
does exist an input that would err, no PPT adversary will be able
to find the erring input with non-negligible probability.
In other words, robust PBCs provide stronger guarantees compared
to normal PBCs. However, we note that batch codes with zero error
are robust PBCs as no erring input exists.
We show this relaxation enables more efficient explicit
constructions.
We will also show later in Section~\ref{sec:reusable_batch_pir} that
robust PBCs may be useful for batch PIR schemes
where hash functions must be made public.

\begin{figure*}[tb]
\centering
\small
\begin{tabular}{| l | c | c | c | c | c |}
        \hline
      \makecell{\bf PBC} & \makecell{\bf Size ($N$)} & \makecell{\bf Buckets ($b$)} &  \makecell{\bf Explicit?}
\\ \hline
    \makecell[l]{Subset~\cite{IKO04}} & $O(n)$ & $q^{O(1)}$ & $\checkmark$ \\ \hline
    \makecell[l]{Expander Graphs~\cite{IKO04}} & $O(n \log n)$ & $O(q)$ & $\times$ \\ \hline
        \makecell[l]{Balbuena Graphs~\cite{rawat2016batch}} & $O(n)$ & $O(q^3)$ & $\checkmark$ \\ \hline\hline
    \makecell[l]{Our Work} & $O(n \cdot (f(n) + \lambda)), f(n) = \omega(\log n)$ & $O(q)$ & $\checkmark$ \\ \hline
\end{tabular}
\caption{A comparison table of $(2^{-\lambda})$-robust PBC constructions with $t = 1$.
\iffull A scheme is $(2^{-\lambda})$-robust if any $\poly(n)$ time adversary can find an erring input with probability at most $2^{-\lambda}$.
\fi}

\label{table:robust_pbc_compare}
\iffull
\else
\vspace*{-5mm}
\fi
\end{figure*}

We present robust PBC constructions
from robust cuckoo hashing. Our robust PBC is the
best explicit construction with $O(q)$ buckets. To our knowledge, all
other robust PBCs come directly from zero-error batch codes.
Furthermore, the most efficient zero-error schemes from expander graphs
are non-explicit.
See Figure~\ref{table:robust_pbc_compare} for more comparison.
At a high level, the PPT adversary is given
the parameters of the scheme (including the hash functions)
and the database. The goal
of the adversary is to produce a subset $Q$ that
cannot be correctly decoded by the scheme.
\iffull
\begin{definition}[Robust Probabilistic Batch Codes]
A $(n, N, q, m, t)$-systematic PBC is $\epsilon$-robust, if for any polynomial time adversary $\calA$, the following holds:
\[
\Pr\left[
R \ne (\db_i)_{i \in Q}: 
\begin{array}{l}
\params \leftarrow \init(1^n)\\
(\db, Q) \leftarrow \calA(1^n, \params)\\
(C_1,\ldots,C_m) \leftarrow \encode(\params, \db)\\
(S_1,\ldots,S_m) \leftarrow \genschedule(\params, Q)\\
R \leftarrow \decode(\params, Q, \{(C_j)_{j \in S_i}\}_{i \in [m]})
\end{array}\right] \le \epsilon.
\]
\end{definition}

\smallskip\noindent{\bf Discussion about Computational Adversaries.}
At first, our usage of computational adversaries may
seem unnecessary. For proving robustness in cuckoo hashing,
we consider PPT adversaries to limit the number
of hash evaluations known to the adversary.
For PBCs, this is already limited by the query universe $[n]$. So, it seems like one could build
a robust PBC against computationally unbounded adversaries.
The difficulty lies in the
random hash functions. If we leverage explicit random hash functions that are $(nk)$-wise independent, we will require
$\Omega(nk)$ storage that increase
the rate logarithmically for databases of single bits.
Therefore, we must construct our random hash functions using
either PRFs or random oracles, which would require computational
guarantees.
\else
See the full version~\cite{cryptoeprint:2022/1455} for more definitions of robust PBCs.
\fi

For our construction, we will simply use a
robust cuckoo hashing scheme and follow
the exact same approach as Lemma~\ref{lem:ch_pbc}.
The additional work needed is to show that one can build
robust PBCs using robust cuckoo hashing.
\iffull
\else
The proof of the following may be found in the full version.
\fi

\begin{theorem}
\label{thm:robust_pbc}
%Suppose that one-way functions exist.
For all $1 \le q \le n$ and any function $f(n) = \omega(\log n)$, there exists a $(n, N, q, b, \ell)$-systematic
PBC that is $(2^{-\lambda})$-robust
where $b = O(q)$, $\ell = 1$ and
$N = O\left(n \cdot (f(n) + \lambda)\right)$.
\end{theorem}
\iffull
\begin{proof}
We break down the analysis into two cases
depending on the value of $q$.
We start for larger values of
$q \ge f(n) = \omega(\log n)$.
This construction follows
from using the
robust cuckoo hashing of
Theorem~\ref{thm:robust_k}.
By setting $\epsilon = 2^{-\lambda}$,
we can set $k = \Theta(f(n) + \lambda)$, $b = O(q)$, $\ell = 1$
and $s = 0$
to obtain a cuckoo hashing scheme
that is $(n, \epsilon)$-strongly robust
assuming that $\calH$ is a PRF.
As our cuckoo hashing uses $k$ disjoint tables, we must have
$k \le b$.
Note that
we require that $k \le b$ as we use $k$ disjoint tables
in cuckoo hashing.
We can guarantee that $k \le b$ by ensuring
$k \le q$ as $b \ge q/\ell = q$ as $\ell = 1$.
We can derive the following requirements for
values of $q$:
$$
O(f(n)) = k \le q \implies q = \Omega(f(n)).
$$
Next, we need to show that the robustness of cuckoo hashing
implies the robustness of the PBC. Suppose this PBC
is not robust and a PPT adversary can find a subset $Q$ of size $q$
that cannot be decoded correctly with
probability strictly larger than $\epsilon$.
In other words, this means that this set of $q$ items
cannot be constructed into a cuckoo hash table
according to the current hash functions.
The same PPT adversary can also
find a set of $q$ items that causes a construction failure
in the cuckoo hashing contradicting that the scheme
was robust.
Therefore, we obtain
the desired robust PBC
for this regime of $q \ge f(n)$.

For the case of $q < f(n)$, we use
the straightforward batch code construction with zero
error where each of the $b = q$ entries
stores all $q$ items that also achieves
the desired parameters for the regime.
\end{proof}
\fi

\section{Private Information Retrieval}
\label{sec:pir}
Private information retrieval (PIR)~\cite{CKG98,CG97} is a powerful cryptographic primitive
that considers the setting where a client wishes to retrieve the $i$-th
entry from a server hold an $n$-entry database. For privacy, the server should
not learn the index $i$ that is queried by the client.
In this section, we present improved constructions of
PIR utilizing our new cuckoo hashing instantiations.

\subsection{Single-Query to Batch PIR Reductions}
\label{sec:app_pir}

Batch PIR is an extension of standard PIR where the client
wishes to perform {\em batch queries}. The client
holds a set $Q \subseteq [n]$ of $q$ queries and wishes
to return the $i$-th entry for all $i \in Q$.
\iffull
We present a definition of batch PIR below along with adversarial error.

\begin{tcolorbox}
$\indgame(1^\lambda)$:
\begin{enumerate}
    \item The challenger $\calC$ runs $\params \leftarrow \init(1^\lambda)$.
    \item The adversary $(\db, Q^0, Q^1, S) \leftarrow \calA(\params)$ on input the parameters $\params$ and outputs a state $\st$, the database $\db$,
    two batch queries $Q^0$ and $Q^1$
    and a subset $S \subseteq [s]$
    of at most $s_\calA$ servers to compromise.
    \item The challenge executes $E \leftarrow \encode(\params, \db)$.
    \item The challenger $\calC$ executes $\query(\params, Q^\eta, E)$ and records transcript $\calT_1,\ldots,\calT_s$ for all $s$ servers.
    \item The challenger $\calC$ sends transcripts $\{\calT_x\}_{x \in S}$ to the adversary $\calA$.
    \item The adversary $\calA(\{\calT_x\}_{x \in S})$
    outputs a bit $b$.
\end{enumerate}
\end{tcolorbox}

\begin{definition}[Batch PIR]
\label{def:pir}
A $q$-query batch PIR scheme consists of the following three efficient randomized algorithms:
\begin{itemize}
    \item $\params \leftarrow \init(1^\lambda)$: The initialization algorithm takes the security parameter $\lambda$ and outputs parameters for the scheme.
    \item $E \leftarrow \encode(\params, \db)$: The encoding algorithm is executed by the server to compute an encoding $E$
    of the database $\db$.
    \item $\res \leftarrow \query(\params, Q, E)$: The query algorithm is jointly executed by the client and server where the client receives the parameters and a set $Q = \{i_1,\ldots,i_k\} \subseteq [n]$ of $q$ queries and the server receives the parameters and the encoded database $E$.
\end{itemize}
The scheme has error at most $\epsilon$ if,
for every database $\db \in \{0,1\}^n$ and every query $Q \subset [n]$ such that $|Q| \le q$,
$$
\Pr[\query(\params, Q, E) \ne \{\db_i\}_{i \in Q} : \params \leftarrow \init(1^\lambda), E \leftarrow \encode(\params, \db)] \le \epsilon.
$$
Finally, the scheme is $(s, s_\calA, \delta)$-secure if for all stateful
PPT adversaries $\calA$ that
compromise $s_\calA$ of the $s$ servers and all sufficiently large
databases $\db$,
$$
|p^0_\calA - p^1_\calA| \le \delta(|\db|)
$$
where $p^\eta_\calA$ is defined as the probability
$\calA$ outputs 1 in $\indgame$.
\end{definition}
\else
We present a formal definition in the full version.
\fi

\begin{figure*}[tb]
\centering
\small
\begin{tabular}{| l | c | c | c | c |}
        \hline
      \makecell{\bf Explicit Batch PIR} & \makecell{\bf Computational Time} & \makecell{\bf Queries} &  \makecell{\bf Error}
\\ \hline
\makecell[l]{Subset~\cite{IKO04}} & $O(n)$ & $q^{O(1)}$ & $0$ \\ \hline
\makecell[l]{Balbuena Graphs~\cite{rawat2016batch}} & $O(n)$ & $O(q^3)$ & $0$  \\ \hline
    \makecell[l]{Pung~\cite{angel2016unobservable}} & $4.5n$ & $9q$ & $2^{-20*}$  \\ \hline
    \makecell[l]{3-way Cuckoo Hashing~\cite{angel2018pir}} & $3n$ & $1.5q$ & $2^{-40*}$  \\ \hline
    \hline
    
    \makecell[l]{Our Work} & $O(n \cdot \sqrt{\lambda/\log\log n})$ & $O(q)$ & $2^{-\lambda}$ \\ \hline
\end{tabular}
%\end{adjustbox}
\caption{A comparison table of explicit blackbox single to batch PIR transformations.
The error probability considers queries chosen
independently of the hash functions.
Asterisks (*) denote experimental error probabilities.}
\label{table:batch_pir_compare}
\iffull
\else
\vspace*{-5mm}
\fi
\end{figure*}

We consider
the problem of taking a PIR construction
for a single-query and efficiently transform
it to a batch-query PIR.
The standard way to do this is to utilize a
(probabilistic) batch code to encode
the database to reduce
the problem of a batch PIR query of size $q$
to executing $q$ single-query PIR schemes
(for example, see~\cite{IKO04,angel2018pir}).
To our knowledge, these approaches
result in the most efficient blackbox
transformations that do not make any
other assumptions about the single query PIR scheme.

We present an improved transformation that leverages our
explicit batch codes in Theorem~\ref{thm:pbc}.
Using our cuckoo hashing
based PBC with more hash functions,
we obtain a transformation with
quadratically smaller computational
overhead compared to prior works.
\iffull
\else
The proof can be found in the full version~\cite{cryptoeprint:2022/1455}.
\fi
We point readers to Figure~\ref{table:batch_pir_compare}
for a comparison.

\begin{theorem}
\label{thm:batch_pir}
Suppose there exists a single-query PIR scheme $\Pi$
with communication $c(n)$ and computation $O(n)$. Then,
there exists a batch-query PIR scheme for $q$ queries
with communication $q \cdot c(O(n\sqrt{\lambda/\log \log n}/q + \lambda))$
and computation $O(n \cdot \sqrt{\lambda/\log \log n})$
with error probability $2^{-\lambda}$.
\end{theorem}
\iffull
\begin{proof}
We use the standard approach
of combining our PBC from Theorem~\ref{thm:pbc} with any single-query PIR scheme $\Pi$.
We instantiate the PBC
with parameters $(n, N, q, O(q), 1)$
for an $n$-entry database for performing
a batch query to $q$ entries
where $N = O(n \cdot \sqrt{\lambda/\log\log n})$
that has $2^{-\lambda - 1}$ error.
We will use $\Pi$ to execute a PIR query
into each of the $O(q)$ buckets
to retrieve the necessary entry.
As $\Pi$ uses $O(n)$ computation,
the total computation becomes $O(N) = O(n \cdot \sqrt{\lambda/\log\log n})$.

For communication, we need to bound
the size of each bucket. Recall
the underlying PBC consists of
$O(q/k)$ entries. Within each table,
we throw $n$ balls uniformly at random
into the $O(q/k)$ entries. Therefore,
the expected size of each entry is
$O(nk/q)$. Using standard ``balls-and-bins'' analysis (see~\cite{mitzenmacher2017probability} for example),
we can guarantee that no entry
will contain more than $\max\{O(nk/q), O(\lambda)\}$ entries except with probability $2^{-\lambda-1}$. Therefore,
the communication of each PIR query
can be upper bounded by $c(O(nk/q + \lambda))$ with error probability at most $2^{-\lambda}$.
\end{proof}
\fi

Note, if we plug in any of the asymptotically optimal single-query PIR schemes with $\tilde{O}(\log n)$ communication and $O(n)$ communication, the resulting batch PIR for $q$ queries has
communication $\tilde{O}(q \log (n\lambda))$ and
computation $O(n \sqrt{\lambda/\log\log n})$ that is nearly optimal
except for the $\tilde{O}(\sqrt{\lambda/\log\log n})$ multiplicative factor in computation.

We note that optimal batch
PIR constructions were shown by Groth, Kiayias and Lipmaa~\cite{groth2010multi} by utilizing the properties
of a specific single-query PIR scheme
of Gentry and Ramzan~\cite{gentry2005single}
that is not
a blackbox transformation.
While asymptotically optimal, more recent PIR schemes
based on lattice-based assumptions are more practically efficient (such as~\cite{ali2021communication,mughees2021onionpir,menon2022spiral}).
The most practical batch PIR schemes do make use of the above transformation using PBCs and
state-of-the-art lattice-based PIR constructions (for example, see~\cite{angel2018pir}).

\subsection{Adversarial Error for Re-usable Batch PIR}
\label{sec:reusable_batch_pir}

In the above batch PIR constructions and prior works~\cite{angel2016unobservable,angel2018pir}
that utilize PBCs
instantiated through cuckoo hashing, 
the error probabilities are considered for batch queries chosen
independent of the hash functions.
In practice, this means that the fresh random hash functions are chosen
for each batch query issued by the client to ensure
that query indices are independent.
Unfortunately, this requires the server to constantly generate new
databases for each set of hash functions and, thus,
each query perform by a client.

\begin{figure*}[tb]
\centering
\begin{adjustbox}{width=\textwidth}
\small
\begin{tabular}{| l | c | c | c | c |}
        \hline
      \makecell{\bf Explicit Batch PIR} & \makecell{\bf Computational Time} & \makecell{\bf Queries} &  \makecell{\bf Adversarial Error}
\\ \hline
\makecell[l]{Subset~\cite{IKO04}} & $O(n)$ & $q^{O(1)}$ & $0$ \\ \hline
\makecell[l]{Balbuena Graphs~\cite{rawat2016batch}} & $O(n)$ & $O(q^3)$ & $0$ \\ \hline
    \makecell[l]{Pung~\cite{angel2016unobservable}} & $4.5n$ & $9q$ & $\ge 1/2$ \\ \hline
    \makecell[l]{3-way Cuckoo Hashing~\cite{angel2018pir}} & $3n$ & $1.5q$ & $\ge 1/2$ \\
    \hline\hline
    \makecell[l]{Our Work} & $O(n \cdot (f(n) + \lambda)), f(n) = \omega(\log n)$ & $O(q)$ & $2^{-\lambda}$ \\ \hline
\end{tabular}
\end{adjustbox}
\caption{A comparison table of explicit re-usable batch PIR schemes. \iffull
Adversarial error $\epsilon$ means an adversary running in
$\poly(n)$ time cannot find an erring input except
with probability $\epsilon$.
\fi
}
\label{table:robust_batch_pir_compare}
\vspace*{-5mm}
\end{figure*}

In an ideal setting, we would like for the server
to generate a single database that could be re-used for multiple
batch PIR queries. Ideally, a set of public hash functions are sampled
once and made available to all clients that may issue batch PIR queries.
The server would only need to encode the database according the hash functions
a single time.
Unfortunately, this means that an adversary during the challenge phase may be able to use the hash functions to pick two
batch PIR queries such that only
one of the two batch PIR queries would fail to allocate. For example,
the PPT adversary could choose to employ any of the attacks that
we outline in Theorem~\ref{thm:lb_robust_k}
using knowledge of the hash functions.
If the adversary's view is different for queries that would fail to allocate correctly, the resulting
batch PIR scheme would be insecure.

To our knowledge, we are unaware of any batch PIR scheme
that enable re-usablility.
Prior works~\cite{angel2016unobservable,angel2018pir}
build batch PIR from PBCs with non-robust cuckoo hashing
where adversaries could employ the attack above.
To solve this problem, we introduce the notion of {\em adversarial error} where the PPT adversary can aim to
choose inputs that will cause query errors and/or
failure to encode databases.
By using our robust PBC schemes, we can
guarantee that error remain low even with
databases and batch queries chosen by a PPT adversary.
In other words, we can construct {\em re-usable batch PIR} schemes with low error rates even with adversarially chosen inputs.
We point readers to Figure~\ref{table:robust_batch_pir_compare} for comparisons
with prior works.
\iffull
\begin{definition}[Adversarial Error for Batch PIR]
A batch PIR scheme has adversarial error $\epsilon$ if for
every PPT adversary $\calA$, the following holds:
\[
\Pr\left[
E =\ \perp \lor\ \query(\params, Q, E) \ne (\db_i)_{i \in Q}: 
\begin{array}{l}
\params \leftarrow \init(1^n)\\
(\db, Q) \leftarrow \calA(1^n, \params)\\
E \leftarrow \encode(\params, \db)
\end{array}\right] \le \epsilon.
\]
\end{definition}

\smallskip\noindent{\bf Discussion about Error Definition.}
Our definition differs slightly from standard error definitions used in most algorithms and data structures. Most non-cryptographic works consider a
statistical definition where the above probability
must hold true for every set of queries $Q$.
In our work, we choose a computational definition
where it is hard for a PPT adversary $\calA$ to
even find an erring query if it exists. In general,
we believe this weaker definition suffices because,
in practice, if it is difficult for a PPT adversary $\calA$ to find such an erring query, it will be
very unlikely for the query to be found and executed in practical applications.
\else
See the formal definition
in the full version.
\fi

We follow the same
approach of building batch PIR using a PBC and a single query PIR. The
main difference is that we utilize
a robust PBC. Given a robust PBC,
if there exists a PPT adversary
that can find erring queries, the
same adversary can also
find subsets of items
that cannot be decoded by the
robust PBC. Our approach results in the most
efficient explicit construction with negligible adversarial
error for the regime of $O(q)$ queries.
The next theorem follows in a similar way as 
Theorem~\ref{thm:batch_pir} using our robust PBC from Theorem~\ref{thm:robust_pbc}.

\begin{theorem}
Suppose there exists a single-query PIR scheme $\Pi$
with communication $c(n)$ and computation $O(n)$.
Fix any $k = \omega(\log n) + \lambda$.
Then,
there exists a batch-query PIR scheme for $q$ queries
with communication $q \cdot c(O(nk/q + \lambda))$
and computation $O(nk)$
with adversarial error $2^{-\lambda}$ .
\end{theorem}

\section{Other Applications}
\label{sec:applications}
\noindent{\bf Private Set Intersection.}
PSI considers the problem where two parties have input sets $X$ and $Y$
and wish to compute the intersection $X \cap Y$. PSI rely on cuckoo hashing to enable the two parties
to co-ordinate similar data into buckets.
Prior works~\cite{pinkas2015phasing,chen2018labeled,pinkas2018scalable,cong2021labeled} utilized
experimental evaluation to pick parameters for a cuckoo hashing and chose
$k \in \{2, 3\}$ hash functions with $b = O(n)$ entries of size $\ell = 1$ and no stash, $s = 0$.
Plugging in our constructions, one can obtain
provable failure at the cost of larger asymptotic overhead.

\smallskip\noindent{\bf Encrypted Search.}
\iffull
Encrypted search considers the setting where a data owner
outsources a corpus of documents to be stored by an untrusted server.
The core of encrypted search is the ability to
outsource an index represented as a multi-map of identifiers
to value tuples.
\fi
Recent work~\cite{patel2019mitigating} employed cuckoo hashing for volume-hiding multi-maps
where the goal is to hide the size of value tuples.
Using our scheme with $k = O(1 + \sqrt{\log(1/\epsilon)/\log n})$, $b = O(n)$, $\ell = 1$ and $s = 0$,
we can
reduce the query overhead quadratically (see 
\iffull
Appendix~\ref{app:volume-hiding}
\else
the full version
\fi
for more details).

\smallskip\noindent{\bf Vector Oblivious Linear Evaluation (VOLE).}
Cuckoo hashing is also used in a recent VOLE protocol~\cite{CCS:SGRR19}. At
a very high level, the construction utilizes batch codes in
a similar way as batch PIR. As a result, the improvements to
batch codes in Section~\ref{sec:pbc_pir} can be plugged into
their construction to obtain improvements.

\smallskip\noindent{\bf Batch PIR with Private Preprocessing.}
A recent work~\cite{EC:Yeo23} presented a blackbox construction
for batch PIR with private preprocessing using a standard single-query
PIR with private preprocessing and any batch code. 
\iffull
Using our batch codes from Section~\ref{sec:pbc_pir}, we
immediately obtain improved blackbox reductions
between
batch and single-query PIR in the private preprocessing setting (similar to the improved blackbox reductions between batch and single-query PIR without preprocessing in Section~\ref{sec:pbc_pir}).
\else
Using our batch codes from Section~\ref{sec:pbc_pir}, we
immediately obtain improved blackbox reductions.
\fi

\section{Conclusions}
In this paper, we present new cuckoo hashing constructions
that obtain better trade-offs between query overhead
and failure probabilities.
For any fixed failure probability, the query
complexity of our new schemes are quadratically smaller
than prior constructions.
Furthermore, we define the notion of robust cuckoo hashing
where the adversary has knowledge of the underlying hash functions.
We show that we can extend our schemes with a large number of
hash functions to obtain robustness while prior approaches
cannot be extended except with linear query overhead.
We also present matching lower bounds for all parameters.
Finally, we obtain state-of-the-art
constructions for probabilistic batch codes and
blackbox reductions from
single-query to batch PIR.

\smallskip\noindent{\bf Acknowledgements.}
The author would like to thank Mo (Helen) Zhou for feedback
on earlier manuscripts and Daniel Noble for pointing out an error and fix in a proof.
This research was supported in part by the Algorand Centres of Excellence programme managed by Algorand Foundation. Any opinions, findings, and conclusions or recommendations expressed in this material are solely those of the authors.

{
\bibliographystyle{alpha}
\bibliography{biblio}}

\appendix

\section{Discussion about Cuckoo Hashing Abstraction}
\label{app:discuss}

For most cryptographic applications of cuckoo hashing,
the core privacy guarantee required is that
the resulting cuckoo hashing table for any
pair of input sets are indistinguishable from an adversary. Furthermore, any operations performed
on the cuckoo hashing table must also not reveal
any information about the input sets where the
adversary views the entries that are retrieved for each operation. However,
there are typically restrictions on the operations
that may be performed to make this requirement easier to satisfy. For usages in PIR~\cite{angel2018pir,demmler2018pir,ali2021communication}, PSI~\cite{pinkas2015phasing,chen2017fast,chen2018labeled,pinkas2018efficient,pinkas2020psi} and ORAM~\cite{pinkas2010oblivious,goodrich2011privacy,patel2018panorama,asharov2020optorama,hemenway2021alibi}, it is guaranteed that each
item will be queried at most once.
For SSE~\cite{patel2019mitigating,bossuat2021sse}, an item may be queried multiple times. Some leakage is permitted, but must only reveal information about
queried items.
In particular, there can be no information leakage about items
that are stored in the cuckoo hashing table.
The main reason is that, if certain locations are empty, it may provide the adversary information about the existence and/or absence of items in the input set that will degrade privacy.

In all of the above mentioned applications of cuckoo hashing, the primitives utilized cuckoo hashing in a very restricted way. In particular, they never
utilized insertion of items into the cuckoo hash tables. Furthermore, all queries are performed
in a non-adaptive manner where all possible locations
for any item are retrieved in a single round.
In the following sections, we elaborate
why these two restrictions exist to ensure no
privacy leakage.

\subsection{Static Cuckoo Hashing}

We briefly outline the insertion algorithm
for cuckoo hashing. To insert an item, first check
if any of the $k$ assigned entries contains at
least one empty location. If so, insert the item.
Otherwise, pick one of the items in the $k$ assigned
entries at random to evict. The original item is inserted into the evicted location and the insertion
process is repeated with the evicted item.
Typically, a limit on the number of evictions
is configured before the last item to be evicted
is stored in the overflow stash. If the stash is full (or does not exist), then the insertion is deemed to have failed.

Looking closely into the insertion algorithm,
it can be seen that the behavior of the insertion
depends highly on whether certain locations
are empty or not.
Just by observing the number of evictions, the
adversary can determine whether certain locations are occupied. One way to mitigate this is to
force the insertion algorithm to execute the maximum
number of evictions even if insertion had
succeeded earlier. Even with this modification, there
are significant privacy leaks as inserting each
algorithm is equivalent to querying all possible locations for the item. As a result, the adversary
may be able to correlate inserted items with queried items. In many cases, the
leakage of cuckoo hashing insertion prevents
certain privacy guarantees.
Therefore, most cryptographic applications avoid using insertion for cuckoo hashing table. Instead, they use
an algorithm for constructing tables oblivious to the input set and exclusively query from a static table.
For an example of where cuckoo hashing insertions
are detrimental for privacy,
see Section 3.2 in~\cite{dvh23}.

\subsection{Non-Adaptive Queries}

Next, we discuss our abstraction
of non-adaptive queries that retrieve all
$k\ell + s$ possible locations
for any queried item.
Another option may be to use an adaptive
query that incrementally retrieves locations
that are more likely to contain a queried item.
As an example, non-stash locations are more likely
to hold queried items than stash locations.
One adaptive approach would be to query only 
entries in the main table. The adaptive query algorithm proceeds to retrieve stash locations only when the item is not found in the non-stash locations.
This could be more efficient as stash locations
will be rarely queried.

Unfortunately, this has two downsides. The first
is that it requires multiple roundtrips between
the querier and the party storing the cuckoo table.
Even worse, adaptive querying is detrimental for
privacy as it leaks information about the occupancy of locations.
By observing the number of roundtrips and the last
retrieved entries, an adversary can learn information about the likelihood that certain locations are populated with items.
The standard way to protect against such leakage
is for the querier to simply download all possible
locations for any queried item. In other words,
this is the original non-adaptive query algorithm
that retrieves all $k\ell + s$ locations
that could store the queried item.

\section{Cuckoo Hashing with More Entries}
\label{sec:num_entries}

For completeness, we consider cuckoo hashing with more entries for
small failure probability and robustness. In general, this setting is not
practically feasible as the number of entries will be much larger
than the input set. In particular, for reasonable failure and robustness parameters,
the number of entries will be super-polynomial in the input size (that is tight according to our lower bounds).

\subsection{Negligible Failure}
We will consider the simple setting where $k = 1$ as it suffices
to obtain optimal constructions for all $k \ge 1$.
As we are considering the case where inputs are chosen independent
of the random hash functions, this
corresponds to simply determining if
a collision will occur.
For small enough $\epsilon$, our result matches the $\Omega(1/\epsilon)$
lower bound that we prove later.

\begin{theorem}
\label{thm:b}
Let $k = 1$, $s = 0$, $\ell = 1$ and $b = O(n^2/\epsilon)$. If $\calH$ is a $(nk)$-wise independent hash function, then the cuckoo hashing scheme $\CH(k, b, \ell, s)$ with
a perfect construction algorithm
has construction failure probability at most $\epsilon$.
\end{theorem}
\begin{proof}
The chosen parameters incur an insertion
failure if and only if any two items hashed into the same entry. This probability
is upper bounded by
$$
{n \choose 2} \cdot \frac{1}{b} \le n^2 / b \le \epsilon \implies b = O(n^2/\epsilon)
$$
to complete the proof.
\end{proof}

\subsection{Robustness}
Using a large number of entries turns out to be a straightforward case to
enable robustness. We simply add it for completeness as it will require
super-polynomial in $n$ storage
that is not feasible in most practical settings.
The result below is tight as it matches
the non-robust lower bounds that we show later.

\begin{theorem}
\label{thm:robust_b}
%Suppose that one-way functions exist.
For security parameter $\lambda$ and error $\epsilon = \negl(\lambda)$, let $k = 1$, $s = 0$, $\ell = 1$ and $b = O(1/\epsilon)$, then
the cuckoo hashing scheme $\CH(k, b, \ell, s)$ with
a perfect construction algorithm is $(\lambda, \epsilon)$-strongly robust.
\end{theorem}
\begin{proof}
Let $U$ be the set
of all items sent to the hash function
where $|U| \le f(\lambda)$
for all $f(\lambda) = \poly(\lambda)$.
Once again, we prove that
the probability that
there exists any two
items amongst the $\poly(\lambda)$ items
sent to the hash functions by the
adversary are mapped to the same entry is as follows:
$$
\Pr[\exists u_1 \ne u_2 \in U : H_1(u_1) = H_2(u_2)] \le {|U| \choose 2} \cdot \frac{1}{b} \le \frac{|U|^2}{b}.
$$
As $\epsilon = \negl(\lambda)$ and $|U| \le \poly(\lambda)$, we can set $b = O(1/\epsilon)$
to get that the above is smaller than $\epsilon$.
\end{proof}

\subsection{Lower Bound}

Finally, we present a lower bound that will match
our two constructions above.

\begin{theorem}
Let $k = O(1)$, $\ell = O(1)$, $b \ge n/\ell$ and $s = O(1)$.
The failure probability of $\CH(k, b, \ell, s)$ cuckoo hashing scheme where $\calH$ is a $(k\ell + s + 1)$-wise independent hash function
satisfies the following:
$$
b = \Omega(1/\epsilon).
$$
\end{theorem}
\begin{proof}
We can re-do the proof of Theorem~\ref{lem:err_lb} until
we obtain the inequality $$k^2\ell + ks \ge \log(1/\epsilon)/\log(b/k).$$
By plugging in that $k = O(1)$, $\ell = O(1)$ and $s = O(1)$,
we obtain the desired result of $b = \Omega(1/\epsilon)$.
\end{proof}

Note, if we choose typical parameters of $\epsilon = \negl(n)$,
this matches the $b = O(n^2/\epsilon)$ upper bounds
of the prior constructions.

\section{Lower Bounds for Single Hash Function Settings}
\label{sec:single_hash_function}
The lower bounds in Section~\ref{sec:negl_lb}
also immediately imply results in
the single hash function setting of $k = 1$.
However, it turns out stronger lower
bounds can be achieved using well-known results in the area of balls-and-bins. It has already been
proved that for $b = n$, the number of
items assigned to a single bin will be $\Omega(\log n/\log \log n)$
with probability except $1/n$
(for example, see Lemma 5.12 in~\cite{mitzenmacher2017probability}).
We re-prove the
theorem
for $b = \alpha \cdot n$ for
constant $\alpha > 1$, but
the proof techniques are identical.
Note, this justifies prior works usage of
$k = 2$ for parameter settings
with large entries $\ell$~\cite{minaud2020note} and
stashes $s$~\cite{kirsch2010more}.

\begin{theorem}
\label{thm:single_hash}
Suppose that $k = 1$ and $b = \alpha n$ for some constant $\alpha \ge 1$. If the cuckoo hashing scheme
$\CH(k, b, \ell, s)$ has failure probability at most $\epsilon \le 1/n$, then
$\ell + s = \Omega(\log n / \log\log n)$.
\end{theorem}
\begin{proof}
If we can prove that $t = \Omega(\log n/\log \log n)$ items are hashed
into the same entry with probability at least $\epsilon$, then we immediately complete the proof
as it must be that
the single entry and stash must store all of the $t$ allocated items implying that
$\ell + s = \Omega(\log n / \log\log n)$.

We use the Poisson approximation
of entry sizes where the number
of items in each entry is modelled
using an independent
Poisson variable with mean $n/b = 1/\alpha$. We denote the
event $E$ when there exists
one bin with at least $t$ items.
We note that $E$ is monotonically
increasing in the number of
items $n$. Therefore, if we bound
the probability of $E$ using the Poisson approximation, we lose
only a factor of $2$ when
considering the true balls-and-bins distribution using Corollary 5.11 in~\cite{mitzenmacher2017probability}.

The probability that Poisson
variable with mean $1/\alpha$
is greater than $t$ is at least
$1/(\alpha^t \cdot t! \cdot e^{1/\alpha})$. The probability
that $n$ independent Poisson variables
are all at most $t$ is
$$
\left(1 - \frac{1}{\alpha^t \cdot t! \cdot e^{1/\alpha}}\right)^n \le e^{-n/(\alpha^t \cdot t! \cdot e^{1/\alpha})}.
$$
The above probability must be at most $\epsilon/2$ as we lose a factor of two from the Poisson approximation.
After re-arranging, it can be seen
that
\begin{align*}
e^{-n/(\alpha^t \cdot t! \cdot e^{1/\alpha})} \le \epsilon/2 \implies t\log t = \Omega(\log n + \log\log(1/\epsilon)).
\end{align*}
We note
that $\log\log(1/\epsilon) = O(\log\log n)$, so we only need $t \log t = \Omega(\log n)$.
Therefore, setting $t = \Omega(\log n/ \log\log n)$ suffices to complete
the proof.
\end{proof}

\section{Implications to Other Primitives}
\label{app:encoding}

Recent works have leveraged techniques
from cuckoo hashing to enable efficient encodings
of data such as cuckoo filters~\cite{fan2014cuckoo}
and oblivious key-value stores~\cite{garimella2021oblivious}.
At a high level, both primitives take 
an input set $\{(\id_1,v_1), \ldots, (\id_n,v_n)\}$
of identifiers and values. Furthermore, all $n$ values
are typically random strings.
Similar to cuckoo hashing, each of the items are assigned
to $k \ge 2$ entries according to $k$ hash functions.
For example, the $i$-th item $(\id_i, v_i)$ is assigned to
entries
$H_1(\id_i),\ldots,H_k(\id_i)$.
Then, the goal is to find assignments to all entries such that
the XOR of the $k$ entries will equal $v_i$.
If we let the resulting table be $T$, then we want
to pick $T$ such that
$$
v_i = T[H_1(\id_i)] \oplus \ldots \oplus T[H_k(\id_i)]
$$
for all $i \in [n]$.
For this setting, the failure probability corresponds
to the case that a table $T$ may not exist to satisfy
the input set.

\smallskip\noindent{\bf Lower Bounds.}
We can show that our lower bounds in Theorems~\ref{lem:err_lb} and~\ref{thm:general_lb} for fixed choices of $\ell = 1$ and $s = 0$ also apply to these primitives.
In the proof of these theorems, the failure probability
is lower bounded by computing the probability that some
subset of items $X$ end up being allocated to a subset
of entries $N(X)$ that is strictly smaller than $X$.
That is, $|X| > |N(X)|$. In this case, it is impossible
to store $|X|$ items into strictly less than $|X|$ entries.

The same argument can also be applied to the above
primitives. Here, there is a system of $|X|$ linear equations
with $|N(X)| < |X|$ variables. As the values of the input set are random,
it is impossible to solve the system of linear equations. As a result,
such a table $T$ will not exist.
As a result, we can see that the lower bounds in Theorems~\ref{lem:err_lb} and~\ref{thm:general_lb} with $\ell = 1$ and $s = 0$ also apply
to these primitives.

\smallskip\noindent{\bf Upper Bounds.}
For constructions, we note that the same arguments can be applied
to show that if the underlying cuckoo hashing construction
can allocate the $n$ items, then there must exist a table
$T$ whose entries satisfy the constraints. The main challenge is to find efficient algorithms to find
such a table $T$. To our knowledge, the best algorithm
would be using Gaussian elimination. We leave it as an open
problem to come up with algorithms for computing the table $T$
using our constructions with large numbers of hash functions.

\smallskip\noindent{\bf Robustness with Private Keys.} We also note that Filic {\em et al.}~\cite{filic2022adversarial} studied robust cuckoo filters
in adversarial environments.
However, this work considered a slightly different model than our work. In particular,
it was shown that cuckoo filters may be constructed to be robust assuming that
a private key can be kept secret from the adversary.

\section{Cuckoo Hashing for Volume-Hiding Multi-Maps}
\label{app:volume-hiding}

\noindent{\bf Overview of~\cite{patel2019mitigating}.}
To construct a volume-hiding multi-map, Patel {\em et al.}~\cite{patel2019mitigating} utilize cuckoo hashing in the
following way. For pair of identifier and value tuple, $(\id, \bv = (v_1,\ldots,v_k))$, we convert it into $k$ identifier and value
tuples: $(\id \mid\mid 1, v_1), \ldots, (\id \mid\mid \ell, v_k)$.
We denote $k$ as the volume of the identifier $\id$.
After flattening all pairs of identifier and value tuples into pairs of
identifier and value, the resulting flattened pairs are inserted
into a cuckoo hashing table with $k = 2$, $b = O(n)$, $\ell = 1$
and $s = O(1 + \log(1/\epsilon)/\log n)$.
To query for any identifier $\id$, one simply performs cuckoo hashing
queries for identifiers $\id \mid\mid 1, \ldots, \id \mid\mid \ell$
where $z$ is the maximum size (i.e., maximum volume) of any value tuple.
There are two options for storing the stash. In the original paper,
the stash is stored by the client and checked locally on each query.
This version has $O(z)$ query overhead but $O(1 + \log(1/\epsilon)/\log n)$ client storage. If one wishes to maintain $O(1)$ client storage,
the stash may be encrypted and outsourced to the server. Now, the client storage is $O(1)$ but the query overhead becomes $O(z + \log(1/\epsilon)/\log n)$.

\smallskip\noindent{\bf Modification using Improved Cuckoo Hashing.}
As a modification, we could utilize our cuckoo hashing scheme
with $k = O(1 + \sqrt{\log(1/\epsilon)/\log n})$, $b = O(n)$, $\ell = 1$
and $s = 0$. As there is no stash anymore, the client storage
remains $O(1)$. Instead, the query overhead now becomes
$O(z \cdot \sqrt{\log(1/\epsilon)/\log n})$.
For small maximum volumes such as $z = O(1)$, our construction
has quadratically smaller overhead than the prior construction
in~\cite{patel2019mitigating}
with $O(1)$ client storage.
As a special case, this new construction would be useful when considering
volume-hiding maps where the goal is to simply hide whether an identifier exists in the map. For this problem, the maximum volume is $z = 1$.

\section{Running Time of Construction Algorithms}
\label{sec:running_times}
Throughout this section, we will rely heavily on the
random bipartite graphs that model cuckoo hashing that
we described in Section~\ref{sec:graph}. We point readers back
to that description for more details on the bipartite graph
and the importance of perfect left matchings.

We will consider several different construction algorithms
for constructing cuckoo hashing tables. We will describe each
algorithm and then analyze the construction times.

We quickly recall the most
common usage of cuckoo hashing in cryptographic primitives.
Typically, one party will have an input set of items $X$ and
will construct a cuckoo hashing table based on the set $X$ (or an encrypted version of the set $X$).
In particular, the construction of the cuckoo hashing table
is done locally by the party and the transcript is typically
not revealed to any other parties including the adversary.
Therefore, in this section, we will consider the running time
of the construction algorithm as if it were being performed by a single
party without any privacy considerations. In that case that one is concerned about timing side channels, the construction algorithm may be padded to execute in the worst case running time.

Given the above, our goal is to bound
the time of construction algorithms
where all $n$ items are given
as input. In other words, we are
going to bound the time of inserting
all $n$ items. In contrast, most previous works
considered the time
to bound the insertion of a single item.

For convenience, we consider the most important
setting for cuckoo hashing with a large number of hash functions $k$ where stashes
are empty $s = 0$, entries may store
at most one item $\ell = 1$ and the number
of entries is linear in the
number of items $b = O(n)$. We will
consider an arbitrary number of hash functions where $k$ is a sufficiently large constant.\footnote{We restrict $k$ to be a sufficiently large constant to ensure that both random walks and BFS behave well. For small $k$, these algorithms may behave differently. For example, random walks with $k = 2$ are essentially following a fixed path.}
The analysis can be extended to
consider arbitrary stash sizes $s$, entry sizes $\ell$ and number of entries $b$.

Finally, we note that a lot of our analysis will rely on
prior works that analyze these insertion/construction
algorithms for cuckoo hashing. Most prior works
consider $k = O(1)$. For example, the bounds
for breadth first search in~\cite{fotakis2005space} proved
$O(1)$ insertion time. However, the real bound is $O(\poly(k))$ that is no longer constant for super-constant $k$. In our work, we modify the proofs of prior works
to obtain bounds with respect to arbitrary $k$.

\begin{table*}[tb]
\centering
\begin{adjustbox}{width=\textwidth}
\small
\begin{tabular}{| l | c | c |}
        \hline
      \makecell{\bf Construction Algorithm} & \makecell{\bf Expected Time} & \makecell{\bf Worst Case Time\\
      \bf with $1 - \negl(n)$ Probability}
\\ \hline
    \makecell[l]{Maximum Cardinality\\ Bipartite Matching~\cite{chen2022maximum}} & $\tilde{O}((nk)^{1 + o(1)})$ & $\tilde{O}((nk)^{1 + o(1)})$  \\ \hline
    \makecell[l]{Breadth First Search~\cite{fotakis2005space}} & $O(n \cdot \poly(k))$ &  $O(n^{1+\alpha} \cdot k)$ \\\hline
    \makecell[l]{Breadth First Search, $k = \omega(\log n)$} & $O(nk)$ &  $O(nk)$ \\\hline
     \makecell[l]{Random Walk~\cite{frieze2009analysis,fountoulakis2013insertion,frieze2019insertion,walzer2022insertion}} & $O(nk)$ & $O(nk \cdot \polylog(n)) $\\ \hline
    \makecell[l]{Local Search Allocation~\cite{khosla2013balls}} & $O(nk)$ &  $O(nk \cdot \polylog(n))$ \\\hline
    %\makecell[l]{Local Insertion for $k = \omega(\log n)$} & $O(nk)$ & $O(nk)$ \\\hline
\end{tabular}
\end{adjustbox}
\caption{A comparison table of construction algorithms and their expected and worst
case running times. We consider the most
important setting of sufficiently large $k$, $s = 0$, $\ell = 1$ and $b = O(n)$.}
\label{table:construct_compare}
\end{table*}

\subsection{Maximum Cardinality Bipartite Matching}

Any allocation of the items corresponds
to a matching in the bipartite graph corresponding
to the cuckoo hashing scheme. Therefore, one can re-phrase
the goal of cuckoo hashing allocation as determining whether
the maximum cardinality bipartite matching is as large
as the number of left vertices. 
Therefore, we can use any maximum cardinality
bipartite matching algorithm such as
Ford and Fulkerson~\cite{ford1956maximal}.

The most efficient one to date is due to Chen {\em et al.}~\cite{chen2022maximum}
that runs in $\tilde{O}(m^{1 + o(1)})$ time for
graphs with $m$ edges. In our
setting, we have that $m = nk$ to
get that the worst case running time
is $\tilde{O}((nk)^{1 + o(1)})$.
However, we will analyze algorithms
with better guarantees below.

\subsection{Breadth First Search}

\noindent{\bf Description.}
Next, we consider breadth first search (BFS) that incrementally inserts
a new item to a current left perfect matching.
A very convenient way to maintain perfect left matchings
corresponding to allocations is by using
the direction of edges. All edges outside of the current
matching are directed from left vertices to right vertices (prior works
have used this technique including~\cite{fotakis2005space}).
On the other hand, all edges in the matching are directed from right vertices to left vertices. When adding a new item (i.e., a left vertex), the goal of inserting the item is equivalent to finding an
alternating path from the new vertex to any right vertex
that is currently unoccupied. By the orientation of the edges,
this amounts to simply finding a path from the new left vertex
to any free right vertex.
To obtain a construction algorithm, we can run breadth first search
for all $n$ items to be allocated. Note, this is a perfect construction algorithm as a new item may be inserted if and only
if there exists a path to a free right vertex and BFS will always
find such a path if it exists.

\medskip\noindent{\bf Analysis.}
Recall that breadth first search (BFS)
may be used to incrementally insert
new items into a cuckoo hash table.
At a high level, BFS starts
from the node corresponding to the
item in the cuckoo graph and attempts
to find an augmenting path to an empty
entry node.

We will build upon the proof techniques of Fotakis, Pagh, Sanders and Spirakis~\cite{fotakis2005space}.\footnote{The paper considers
cuckoo hashing with a single shared table. However, their proofs rely on vertex expansion in the cuckoo graph that is better when considering multiple disjoint tables.}
At a high level, they proved
that BFS takes $O(\poly(k))$ expected time and $O(n^\alpha \cdot k)$
worst case time for some constant $0 < \alpha < 1$ except with probability $(2^{-\Omega(n)} + n^{-\Omega(k)})$.\footnote{The original theorem was meant to apply for all values of $k \ge 8$. As a result, they could only obtain worst case times except with $1 - n^{-\Omega(1)}$ probability.}
We formally present
the theorem below.

\begin{theorem}[\cite{fotakis2005space}]
\label{thm:bfs_prior}
Let $k = \omega(1)$, $s = 0$, $\ell = 1$ and $b = O(n)$.
For cuckoo hashing $\CH(k, b, \ell, s)$
with a BFS construction algorithm
where $\calH$ is $(nk)$-wise independent,
then a single item will be inserted using a random walk
in expected time $O(\poly(k))$ and worst case time $O(n^\alpha \cdot k)$ for some constant $0 < \alpha < 1$ except with probability
$1/2^{\Omega(n)} + 1/n^{\Omega(k)}$.
\end{theorem}

\begin{theorem}
\label{thm:bfs}
Let $k = \omega(1)$, $s = 0$, $\ell = 1$ and $b = O(n)$.
For cuckoo hashing $\CH(k, b, \ell, s)$
with a BFS construction algorithm
where $\calH$ is $(nk)$-wise independent,
the following holds:
\begin{itemize}
    \item The expected construction time is $O(n \cdot \poly(k))$.
    \item The worst case construction time is $O(n^{1+\alpha} \cdot k)$ except
    with negligible probability.
\end{itemize}
\end{theorem}
\begin{proof}
The expected time follows trivially
from Theorem~\ref{thm:bfs_prior}.
As $k = \omega(1)$, the worst case
time of any insertion is $O(n^\alpha \cdot k)$
from Theorem~\ref{thm:bfs_prior} meaning
the worst case time is $O(n^{1+\alpha} \cdot k)$.
\end{proof}

To our knowledge, it remains an open problem to prove
worst case bounds for BFS with high probability that are
sub-polynomial.

We note that we can prove better bounds
for BFS for larger values of $k = \omega(\log n)$.
In particular, we show that BFS will terminate
after the first step except with negligible probability. As a result, we can prove optimal
worst case construction times of $O(nk)$.

\begin{theorem}
Let $k = \omega(\log n)$, $s = 0$, $\ell = 1$ and $b = O(n)$.
For cuckoo hashing $\CH(k, b, \ell, s)$
with a BFS construction algorithm
where $\calH$ is $(nk)$-wise independent,
the following holds:
\begin{itemize}
    \item The expected construction time is $O(nk)$.
    \item The worst case construction time is $O(nk)$ except
    with negligible probability.
\end{itemize}
\end{theorem}
\begin{proof}
Consider the BFS algorithm at the first step.
Assuming that $b = \alpha n$ for some constant $\alpha > 1$, we know that the hash
table is only half full. Suppose the number
of entries in the first table is $n_1$, the number of entries
in the second table is $n_2$ and so forth
such that $n_1+\ldots +n_k \le n$.
The probability that the $k$ different entries are
all occupied is at most $(n_1 /(\alpha n/k)) \cdots (n_k / (\alpha n/k))$
that is maximized by setting $n_1 = \ldots = n_k = k/n$
to get an upper bound on the probability of $(1/\alpha)^k$. As $k = \omega(\log n)$ and $\alpha$ is a constant strictly greater than $1$, the probability that the first step of BFS fails to insert the item is negligible. As a result, we get that the worst case time to insert a single item is $O(k)$
except with negligible probability.
This implies the expected and worst
case time (except with negligible probability)
to insert $n$ items is $O(nk)$.
\end{proof}

\subsection{Random Walk}

\noindent{\bf Description.}
The last insertion algorithm 
considers random walks. Using the same
idea of find alternating paths
from BFS, we could instead try to find such a path using a random walk.
In general, random walks are considered
more efficient as they require less memory usage compared to BFS as one
does not need to remember the prior
paths. To do this, as the random walk
visits entries, it will evict items
from full entries and only remember that item that it will try to place next.

The main challenge is that a
random walk can, in theory, run
forever without terminating even if an
alternating path exists. In the past,
cuckoo hashing with random walks
will bound the length of the random
walk before simply terminating and putting the item in the stash. However,
this does not work for us if
we wish to obtain a perfect construction algorithm.

Instead, we can obtain a perfect construction algorithm using a different approach. In particular, we
can let the random walk algorithm
execute until trying a path
of length $O(n)$. If such a path has not yet been found, we can
instead execute a BFS to find the
alternating path that also has worst
case $O(nk)$ overhead.

\medskip\noindent{\bf Worst Case Analysis.}
We use the results of
Fountoulakis, Panagiotou and Steger~\cite{fountoulakis2013insertion}
that consider the running time of inserting a single item using a random walk.
We could also use the results of Frieze, Melsted and Mitzenmacher~\cite{frieze2009analysis}
as we ignore additional log factors
and consider larger values of $k$.
At a high level, they show
a random walk must traverse $\polylog(n)$
nodes before finding an augmenting path.

\begin{theorem}[\cite{fountoulakis2013insertion}]
Let $k \ge 3$, $s = 0$, $\ell = 1$ and $b = O(n)$.
For cuckoo hashing $\CH(k, b, \ell, s)$
where $\calH$ is $(nk)$-wise independent,
then a single item will be inserted using a random walk
in worst case time $O(\polylog(n))$ except with probability
$1/n$.\footnote{In the original paper~\cite{fountoulakis2013insertion},
the authors showed this is true except with probability $1/n^\alpha$ for $\alpha > 0$. However, by increasing the running time by poly-logarithmic factors, one
can drive the probability down to $1/n$.}
\end{theorem}

Using this result, we can immediately derive
the following theorem that bounds
the expected and worst case running times
of random walks.

\begin{theorem}
Let $k \ge 3$, $s = 0$, $\ell = 1$ and $b = O(n)$.
For cuckoo hashing $\CH(k, b, \ell, s)$
with a random walk construction algorithm
where $\calH$ is $(nk)$-wise independent,
then the worst case construction time is $O(nk \cdot \polylog(n))$ except
    with negligible probability.
\end{theorem}
\begin{proof}
We note that we can
consider random variables $X_i = 1$ if and
only if the $i$-th item that is inserted
exceeds time $O(\polylog(n))$.
We know that $\Pr[X_i = 1] \le 1 - 1/n$.
If we let $X = X_1 + \ldots + X_n$, we know
that $\Pr[X > \log^2 n] \le \negl(n)$ by
Chernoff Bounds. Therefore, the total
construction time except with negligible
probability is $
(n - X) \cdot O(\polylog(n)) + X \cdot O(nk) = O(nk \cdot \polylog(n))$ as $X \le \log^2 n$
except with negligible probability.
\end{proof}

\noindent{\bf Expected Analysis.}
We note the same papers~\cite{frieze2009analysis,fountoulakis2013insertion} that we relied upon above
proved a $\polylog(n)$ expected time bound
for random walks. Instead, we can move to more recent
works~\cite{frieze2019insertion,walzer2022insertion} that proved constant bounds for the
expected times of random walks. Again, they proved
bounds for constant $k$, but one can re-interpret their
results for arbitrary $k$.

\begin{theorem}[\cite{walzer2022insertion}]
Let $k \ge 3$, $s = 0$, $\ell = 1$ and $b = O(n)$.
For cuckoo hashing $\CH(k, b, \ell, s)$ with a random
walk construction algorithm where $\calH$ is $(nk)$-wise independent, then  a single item will be inserted using a random walk in expected time $O(k)$.
\end{theorem}

By linearity of expectation, the expected construction time
for random walks immediately follows as $O(nk)$.

\subsection{Local Search Allocation}

\noindent{\bf Description.}
Khosla~\cite{khosla2013balls} presented a construction
algorithm for the case of $k \ge 3$, $\ell = 1$, $s = 0$ and
$b = O(n)$. Unlike BFS and random walks, this is a construction algorithm that takes the entire input at once and aims
to create a hash table.

At a high level, each of the $b$
entries maintains an integer label. Initially,
all labels are $0$. Whenever an item is inserted,
it will search amongst its $k$ options to find the minimum
label entry. The item will be placed into that entry. Furthermore, the label of that entry is increased to be
one more than the minimum amongst the other $k - 1$ options for
that item. If that entry was already occupied, we repeat
this algorithm by evicting the original item in the entry.
This is considered a construction algorithm as one must
maintain the labels throughout the construction. If one is content with storing all labels, it is possible to convert
this algorithm into an insertion algorithm as well.

In comparison with BFS, we note that both algorithms require
linear storage. This algorithm provides stronger expected and worst-case
guarantees compared to BFS. Compared to random walk,
this algorithm requires more storage but can still provide stronger
worst case guarantees. Furthermore, experiments in the original
paper~\cite{khosla2013balls} showed that the algorithm
outperforms random walks in terms of running time.

\medskip\noindent{\bf Analysis.} In the paper~\cite{khosla2013balls}, it was proven that the construction
of this algorithm is $O(nk)$ except with $1/\poly(n)$ probability. As a result, we can immediately get that the expected running time is also $O(nk)$. Similarly, it was shown
that the maximum contribution to the total running time by any vertex (i.e., the label) is
at most $O(\log n)$ with $1/\poly(n)$ probability. As a result, we can obtain a worst case bound as well by showing
that at most $\polylog(n)$
such vertices will exceed $O(\log n)$ except with negligible probability.

\subsection{Robust Construction Algorithms}
In the prior sections, we only analyze
the running times for construction
algorithms with the assumption
that items to be constructed
are chosen independent of the hash functions.
This is necessary as a computationally
unbounded adversary can find
sequences that incur the worst case construction times for random walks and BFS.
In fact, such powerful adversaries can find
sets of items that will cause the construction algorithm to fail.
If one wishes to consider cuckoo hashing
construction with respect to compromised
randomness known by an adversary,
then the best option is to use
deterministic algorithms. To our knowledge,
the best deterministic algorithm
for maximum bipartite matching is
$(nk)^{4/3 + o(1)}$ by Kathuria, Liu and Sidford~\cite{kathuria2022unit}.

\end{document}